%% file: main.tex
\documentclass[11pt]{article}

\usepackage[margin=1in]{geometry}
\usepackage{iftex}
\usepackage[T1]{fontenc}
\ifPDFTeX\usepackage[utf8]{inputenc}\fi
\usepackage{lmodern}
\usepackage{microtype}
\usepackage{amsmath,amssymb,amsthm,mathtools}
\usepackage{booktabs,array,multirow}
\usepackage{graphicx}
\usepackage{capt-of}
\usepackage[table]{xcolor}
\definecolor{floorshade}{gray}{0.90}
\usepackage{enumitem}
\usepackage{algorithm}
\usepackage{algpseudocode}
\usepackage[numbers,sort&compress]{natbib}
\usepackage{authblk}
\usepackage[colorlinks=true,linkcolor=blue!60!black,citecolor=green!45!black,urlcolor=blue!60!black]{hyperref}
\usepackage{xurl}
\usepackage[capitalize,noabbrev]{cleveref}
\hypersetup{pdftitle={Near-Exact Computation of Independent Chip Model Placement Probabilities for Thousands of Players}, pdfauthor={Wataru Inariba}}

\input{macros}

\InputIfFileExists{tables/accuracy_summary}{}{}

\InputIfFileExists{tables/mc_summary}{}{}
\InputIfFileExists{tables/speedup_summary}{}{}
\InputIfFileExists{tables/sweep_summary}{}{}
\providecommand{\SpeedupEightMin}{\placeholder{?}}
\providecommand{\SpeedupEightMax}{\placeholder{?}}
\providecommand{\NumInstances}{\placeholder{?}}
\providecommand{\NumCrossChecked}{\placeholder{?}}
\providecommand{\NumReferencePairs}{\placeholder{?}}
\providecommand{\NumStacks}{\placeholder{?}}
\providecommand{\NumPrizeStructures}{\placeholder{?}}
\providecommand{\WorstRelV}{\placeholder{?}}
\providecommand{\WorstAbsQ}{\placeholder{?}}
\providecommand{\WorstRelQten}{\placeholder{?}}
\providecommand{\MaxN}{\placeholder{?}}
\providecommand{\WorstSumDev}{\placeholder{?}}
\providecommand{\WorstColDev}{\placeholder{?}}
\providecommand{\WorstRelF}{\placeholder{?}}
\providecommand{\McPerSampleFourK}{\placeholder{?}}
\providecommand{\PredRatioMin}{\placeholder{?}}
\providecommand{\PredRatioMax}{\placeholder{?}}

\title{Near-Exact Computation of Independent Chip Model\texorpdfstring{\\}{ }Placement Probabilities for Thousands of Players}
\author{Wataru Inariba\thanks{\texttt{wataru@gtowizard.com}}}
\affil{GTO Wizard}
\date{October 2026}

\begin{document}
\maketitle

\input{sections/00-abstract}
\input{sections/01-introduction}
\input{sections/02-background}
\input{sections/03-quadrature}
\input{sections/04-allplayers}
\input{sections/05-results}
\input{sections/06-related}
\input{sections/07-discussion}

\section*{Acknowledgments}
I would like to thank Philippe Beardsell, for entrusting me with the problem that began this work
and for tracing its connections to neighboring fields such as Plackett--Luce models; Adam Kelly,
for insightful discussions of the algorithm in its early form; Christopher Solinas, for endorsing
the submission to arXiv; and my colleagues at GTO Wizard, for the opportunity to pursue this work
and to publish it. The manuscript was prepared with the assistance of large language models;
I reviewed the resulting text and take full responsibility for its contents.

\bibliographystyle{plainnat}
\bibliography{refs}

\appendix
\input{sections/A-appendix}

\end{document}

%% file: macros.tex
\newcommand{\todo}[1]{{\color{red!70!black}\textbf{[TODO: #1]}}}

\newcommand{\DEICM}{\mbox{DE-ICM}}
\newcommand{\MPL}{\mbox{MPL}}
\newcommand{\LAQI}{\mbox{LAQI}}
\newcommand{\ICM}{ICM}

\DeclareMathOperator{\Prob}{P}
\DeclareMathOperator{\Ex}{E}
\DeclareMathOperator{\Exp}{Exp}

\newcommand{\ind}[1]{\mathbf{1}[#1]}

\newcommand{\stack}[1]{s_{#1}}            
\newcommand{\nstack}[1]{\tilde{s}_{#1}}   
\newcommand{\stacksum}{S}                 
\newcommand{\prize}[1]{\pi_{#1}}          
\newcommand{\Q}[2]{Q_{#1,#2}}             
\newcommand{\V}[1]{V_{#1}}                
\newcommand{\clock}[1]{T_{#1}}            
\newcommand{\behind}[1]{q_{#1}}           
\newcommand{\pgf}{G}                      
\newcommand{\coef}[1]{C_{#1}}             
\newcommand{\loo}[2]{C^{\setminus #1}_{#2}}      
\newcommand{\scaledloo}[2]{E^{\setminus #1}_{#2}}  
\newcommand{\loopgf}[1]{\pgf^{\setminus #1}}     
\newcommand{\mult}[1]{\alpha_{#1}}        

\newcommand{\dexp}{\varphi}               
\newcommand{\gmax}{g_{\max}}             
\newcommand{\step}{h}                     
\newcommand{\xL}{x_{\mathrm{L}}}
\newcommand{\xR}{x_{\mathrm{R}}}
\newcommand{\nodes}{M}                    
\newcommand{\ueps}{u}                     

\theoremstyle{plain}
\newtheorem{theorem}{Theorem}[section]
\newtheorem{lemma}[theorem]{Lemma}
\newtheorem{proposition}[theorem]{Proposition}
\newtheorem{corollary}[theorem]{Corollary}
\theoremstyle{definition}

\newtheorem{remark}[theorem]{Remark}

%% file: tables/accuracy_summary.tex
\newcommand{\NumInstances}{1,511}
\newcommand{\NumCrossChecked}{141}
\newcommand{\NumReferencePairs}{1,662}
\newcommand{\NumStacks}{189}
\newcommand{\NumPrizeStructures}{10}
\newcommand{\WorstRelV}{$4.0 \times 10^{-14}$}
\newcommand{\WorstAbsQ}{$1.1 \times 10^{-14}$}
\newcommand{\WorstRelQten}{$1.4 \times 10^{-7}$}
\newcommand{\MaxN}{4,000}
\newcommand{\WorstSumDev}{$2.2 \times 10^{-14}$}
\newcommand{\WorstRelF}{$6.0 \times 10^{-14}$}
\newcommand{\WorstColDev}{$4.7 \times 10^{-13}$}

%% file: tables/mc_summary.tex
\newcommand{\McPerSampleFourK}{167}

%% file: tables/speedup_summary.tex
\newcommand{\SpeedupEightMin}{4.8}
\newcommand{\SpeedupEightMax}{6.0}

%% file: tables/sweep_summary.tex
\newcommand{\PredRatioMin}{0.70}
\newcommand{\PredRatioMax}{1.09}

%% file: sections/00-abstract.tex
\begin{abstract}
The Independent Chip Model (\ICM{}) converts the chip stacks of the players remaining in a poker
tournament into finishing-place probabilities and prize equities. It is the standard model in
tournament solvers, but its definition considers all possible finishing orders, and exact
computation has been regarded as intractable for large fields. The same mathematics appears in other
fields; for instance, the \ICM{} is a Plackett--Luce ranking model with stacks as weights.
We present \DEICM{}, a deterministic algorithm that computes the placement probabilities of all $n$
players for all paid places in $O(\nodes n^{2})$ time, where the number $\nodes$ of quadrature nodes
is fixed for the admissible inputs and the target accuracy. Its numerical scheme allows the target
to be set close to the accuracy of double-precision arithmetic. The algorithm evaluates a classical
integral representation in which, conditional on a player's exponential clock, the number of players
ahead is Poisson-binomial. A double-exponential quadrature on a single grid evaluates every integral;
its truncation errors have closed-form bounds, and the step size, which controls the discretization
error, follows from the field size. A numerically stable deconvolution then recovers each player's
leave-one-out coefficients in $O(n)$ time from one shared product.
On exactly solvable instances with up to \MaxN{} players, the worst observed relative error of an
equity is \WorstRelV{} and the worst absolute error of a placement probability \WorstAbsQ{}.
The full placement matrix for 1,000 players takes 0.2 seconds on one CPU core.
\end{abstract}

%% file: sections/01-introduction.tex
\section{Introduction}
\label{sec:intro}

In a poker tournament, chips are not money. Prizes are paid by finishing position, so the
value of a stack depends on the stacks of every remaining player and on the prize structure.
The Independent Chip Model (\ICM{}) is the standard way to convert stacks into value: it defines
a distribution over finishing orders by repeatedly drawing the next finisher with probability
proportional to the remaining stacks \citep{harville1973,malmuth1987}.
In a tournament solver that approximates equilibrium strategies \citep{zinkevich2007}, the
equities under the \ICM{} serve as the payoffs at the end of a hand, either directly or as the
initial values of an iteration that refines them \citep{ganzfried2008}.

The defining recursion enumerates all finishing orders. For $n$ players, direct enumeration costs
$O(n!)$ time, and dynamic programming over subsets of players costs $O(n 2^{n})$ time and
$O(2^{n})$ space; neither is practical beyond about $30$ players.
Larger fields have relied on Monte Carlo simulation of finishing orders \citep{streib2011}
or on other approximation methods, such as the undisclosed method used by \citet{hrc2020}.
Since a tournament solver may require thousands of ICM evaluations, and small differences in equity
can affect the resulting strategies, we aim for fast, deterministic evaluation whose accuracy can be
tuned through an explicit error budget.

This paper describes \DEICM{} (\cref{alg:deicm}), which computes the matrix of placement
probabilities $\Q{i}{r}$ for every player $i$ and every paid place $r \le m$, with accuracy that can
approach that of double-precision arithmetic. It runs in $O(\nodes n^{2})$ time, where the number
$\nodes$ of quadrature nodes is derived once from the field size, stack-ratio bound and target
accuracy, and is shared across admissible instances.
The matrix rather than the equities is the natural output because the equity under any prize
structure is a linear functional of it. Its entries and partial sums, such as the probability of
finishing in the money, are also of interest in themselves.

The starting point is a classical integral representation (\cref{thm:representation}).
The \ICM{} is a Plackett--Luce model \citep{luce1959,plackett1975} with the stacks as weights, and
in the exponential-race representation of that model, players are ordered by increasing values of
independent exponential clocks with rates given by the stacks.
Given one player's clock value, each other player independently finishes ahead or behind, so the
number ahead has a Poisson-binomial distribution. Integrating these conditional probabilities
against the exponential density of that player's clock gives the placement probabilities.
This representation arises in the analysis of the move-to-front rule for self-organizing lists
\citep{flajolet1992,fill1996holst} and was recently rediscovered for propensity estimation in
learning to rank \citep{knyazev2026}.

\paragraph{Quadrature.}
General-purpose quadrature typically adapts its quadrature nodes to the relevant ranges of clock
values, which vary with the input stack distribution.
In contrast, we normalize the stacks and evaluate all integrals on a fixed grid using a
double-exponential (DE) transformation followed by the trapezoidal rule.
We assume throughout that the largest and the smallest stack differ by at most a factor of
$\rho = 10^{9}$, a conservative bound for the tournament settings of interest.
The truncation errors have closed-form bounds, and the step $\step$ follows from a prediction of the
discretization error that is checked against measurements (\cref{sec:quadrature}).

\paragraph{Leave-one-out.}
The Poisson-binomial generating function for each player excludes the player's own factor.
At each quadrature node, recomputing the product separately for every player costs $O(n^{2} m)$.
Dividing the full product by each player's linear factor costs $O(n)$ per player, or $O(n^2)$ in
total, but previous work avoided this approach because of numerical instability \citep{knyazev2026}.
Here the division is stable once its direction is chosen by comparing the player's conditional
probability to a threshold, and its error satisfies an absolute bound. Within that bound,
rounding can still leave a residual that turns into an appreciable relative error after quadrature,
so we remove it by clamping each step to the bounds that the exact recurrence satisfies (\cref{sec:allplayers}).

\paragraph{Production use.}
An earlier version of the method, with the same integral representation and quadrature on a fixed
grid based on the exp-sinh transformation, has been in production use at GTO Wizard since 2024,
inside a tournament solver for fields of up to 4,096 players. The direction-switched leave-one-out
deconvolution has been in use since then in an ICM calculator that returns the equities of all
players. A GTO Wizard blog post \citep{tombos2024} reported the existence and accuracy of this
earlier version without describing the method.
This paper adds an explicit error budget and parameter selection for a double-exponential rule
suited to exponentially decaying integrands (\cref{sec:quadrature}), a stability analysis and a
clamp for the leave-one-out recurrence (\cref{sec:allplayers}), and a validation suite at the level
of double-precision arithmetic (\cref{sec:results}), along with comparisons with marginalized
Plackett--Luce \citep{knyazev2026} and Log-Age Quadrature ICM \citep{wolters2026}
(\cref{sec:related-pl}). The parameters and the experiments reported here were developed for this
paper and are independent of the product.

\paragraph{Recent related work.}
Two recent works, \MPL{} \citep{knyazev2026} and \LAQI{} \citep{wolters2026}, arrive at
the same integral representation and present numerical frameworks for evaluating it.
They show that Plackett--Luce propensities and ICM equities, respectively, can be computed
deterministically for large fields, with much higher accuracy than sampling.
\LAQI{} cites the above GTO Wizard blog post as its starting point.
\DEICM{} differs in deriving its grid from an error budget rather than from the instance, in
computing the placement matrix at $O(n^{2})$ operations per node, and in its validation against
exact references for up to $4{,}000$ players (\cref{tab:comparison} in \cref{sec:related-pl}).

\medskip\medskip
\cref{sec:results} validates the algorithm against families of exactly solvable instances that
scale to \MaxN{} players and reports runtimes. \cref{sec:related} relates the method to earlier
work, and \cref{sec:discussion} concludes the paper.

%% file: sections/02-background.tex
\section{Background and the integral representation}
\label{sec:background}

\subsection{The Independent Chip Model}
\label{sec:icm}

A field of $n$ players has stacks $\stack{1}, \dots, \stack{n} > 0$ with total
$\stacksum = \sum_{i} \stack{i}$. Prizes $\prize{1} \ge \prize{2} \ge \dots \ge \prize{m} > 0$
are paid to the first $m \le n$ finishers, and $\prize{r} = 0$ for $r > m$. The \ICM{}
\citep{harville1973,malmuth1987} assigns to a finishing order $\sigma = (\sigma_1, \dots,
\sigma_n)$, where $\sigma_r$ is the player finishing in place $r$, the probability
\begin{equation}
  \Prob(\sigma) \;=\; \prod_{r=1}^{n} \frac{\stack{\sigma_r}}{\stacksum - \sum_{r' < r} \stack{\sigma_{r'}}} .
  \label{eq:icm}
\end{equation}
That is, the winner is drawn with probability proportional to stack, then the runner-up is
drawn among the remaining players with probability proportional to stack, and so on. This is
the Plackett--Luce distribution \citep{luce1959,plackett1975} with weights $\stack{i}$. Two
quantities are of interest: the \emph{placement probabilities}
$\Q{i}{r} = \Prob(\sigma_r = i)$, which form a doubly stochastic $n \times n$ matrix, and the
\emph{equities} $\V{i} = \sum_{r=1}^{m} \prize{r} \Q{i}{r}$. Because \cref{eq:icm} is invariant
under $\stack{i} \mapsto \lambda\,\stack{i}$ for $\lambda > 0$, only the ratios of stacks matter; we use
this freedom in \cref{sec:normalization}. \Cref{tab:notation} collects the notation.

\subsection{The exponential race}
\label{sec:race}

\begin{lemma}[Exponential race]
  \label[lemma]{lem:race}
  Let $\clock{1}, \dots, \clock{n}$ be independent with $\clock{j} \sim \Exp(\stack{j})$, and
  let $\sigma$ order the players by increasing clock, so that $\sigma_1$ has the smallest clock.
  Then $\sigma$ has the \ICM{} distribution given by \cref{eq:icm}.
\end{lemma}

\begin{proof}
  The minimum of independent exponentials is attained by player $i$ with probability
  $\stack{i} / \stacksum$. Conditional on player $i$ having the minimum at time $t$, the residual
  times $\clock{j} - t$, $j \neq i$, are again independent exponentials with rates $\stack{j}$
  by the memoryless property. Induction over the remaining players gives \cref{eq:icm}.
\end{proof}

The lemma gives a simple way to sample an \ICM{} finishing order: draw the clocks and sort them
\citep{streib2011}. It also corresponds to the Gumbel construction of the Plackett--Luce
distribution \citep{yellott1977,kool2019}, since $-\log \clock{j}$ is a Gumbel variable with
location $\log \stack{j}$, and to the continuous-time description of the move-to-front rule
(\cref{sec:related-mtf}).

\begin{table}[t]
  \centering
  \footnotesize
  \setlength{\tabcolsep}{4pt}
  \caption{Notation.}
  \label{tab:notation}
  \begin{tabular}{@{}lll@{}}
    \toprule
    Symbol & Meaning & Notes \\
    \midrule
    $n$, $m$ & players, paid places ($m \le n$) & \\
    $\stack{i}$, $\nstack{i}$, $c$ & stack of player $i$, normalized stack, headroom & $\nstack{i} = \stack{i} / (c\, s_{\max})$, $c = 20$ \\
    $\rho$ & stack-ratio bound, $s_{\max} / s_{\min} \le \rho$ & $\rho = 10^{9}$ \\
    $\prize{r}$ & prize for place $r$ & non-increasing, $\prize{r} = 0$ for $r > m$ \\
    $\Q{i}{r}$, $\V{i}$ & placement probability, equity & $\V{i} = \sum_r \prize{r} \Q{i}{r}$ \\
    $\clock{i} \sim \Exp(\stack{i})$ & exponential clock of player $i$ & smaller clock, earlier finish \\
    $t$ & value of the conditioned clock & \\
    $\behind{j}(t) = e^{-\stack{j} t}$ & $\Prob(\clock{j} > t)$ & $j$ finishes behind the conditioned player \\
    $\coef{r}(t)$ & coefficient of $z^{r}$ in $\prod_{j} (\behind{j} + (1-\behind{j}) z)$ & \\
    $\loo{i}{r}(t)$ & coefficient of $z^{r}$ in the product over $j \ne i$ & leave-one-out; $r$ players ahead, place $r+1$ \\
    $\dexp(x) = \exp(x - e^{-x})$ & exp-exp transformation, $t = \dexp(x)$ & \cref{sec:rule} \\
    $\step$, $[\xL, \xR)$, $\nodes$ & grid step, window, number of quadrature nodes & \\
    $\ueps$ & unit roundoff & $2^{-53} \approx 1.11 \times 10^{-16}$ in double precision \\
    \bottomrule
  \end{tabular}
\end{table}

\subsection{Integral representation of placement probabilities}
\label{sec:representation}

Fix a player $i$ and condition on $\clock{i} = t$. Player $j \neq i$ finishes ahead of $i$
exactly when $\clock{j} < t$, an event of probability $1 - \behind{j}(t)$ with $\behind{j}(t) =
e^{-\stack{j} t}$, and these events are independent across $j$. The number of players finishing
ahead of $i$ is therefore Poisson-binomial with success probabilities $1 - \behind{j}(t)$, $j
\neq i$, whose probability generating function is
\begin{equation}
  \loopgf{i}(z; t) \;=\; \prod_{j \neq i} \bigl( \behind{j}(t) + (1 - \behind{j}(t))\, z \bigr)
  \;=\; \sum_{r=0}^{n-1} \loo{i}{r}(t)\, z^{r} .
    \label{eq:loo-pgf}
\end{equation}
The coefficient $\loo{i}{r}$ is the probability that exactly $r$ players finish ahead of $i$, so
it corresponds to place $r + 1$; this offset of one between the index of a coefficient and a
place is kept throughout. Expanding the product gives all conditional placement probabilities of
player $i$ at fixed $t$ in $O(n^{2})$ operations, without enumerating finishing orders.

\begin{theorem}[Integral representation]
  \label{thm:representation}
  For every player $i$ and every place $r+1$ with $0 \le r \le n-1$,
  \begin{equation}
    \Q{i}{r+1} \;=\; \int_{0}^{\infty} \stack{i}\, e^{-\stack{i} t}\; \loo{i}{r}(t)\; \mathrm{d}t .
    \label{eq:representation}
  \end{equation}
\end{theorem}

\begin{proof}
  By \cref{lem:race}, $i$ finishes in place $r+1$ exactly when $r$ clocks are smaller than
  $\clock{i}$. Conditional on $\clock{i} = t$, the number of smaller clocks is Poisson-binomial
  with the generating function in \cref{eq:loo-pgf}, so the probability that it equals $r$ is
  $\loo{i}{r}(t)$. Integrating against the density $\stack{i} e^{-\stack{i} t}$ of $\clock{i}$
  gives \cref{eq:representation}.
\end{proof}

\begin{corollary}
  \label[corollary]{cor:equity}
  $\V{i} = \int_{0}^{\infty} \stack{i} e^{-\stack{i} t} \sum_{r=0}^{m-1} \prize{r+1}
  \loo{i}{r}(t)\, \mathrm{d}t$.
\end{corollary}

%% file: sections/03-quadrature.tex
\section{Quadrature on a fixed double-exponential grid}
\label{sec:quadrature}

The goal of this section is a single set of quadrature nodes $t_0, \dots, t_{\nodes-1}$ and weights
$w_0, \dots, w_{\nodes-1}$ that evaluates every integral of \cref{thm:representation}, for every
admissible instance, to a target accuracy, so that the work per node can be shared across players.

\subsection{Admissible inputs and normalization}
\label{sec:normalization}

We admit stack vectors whose largest and smallest stacks differ by at most a fixed ratio,
\begin{equation}
  s_{\max} / s_{\min} \;\le\; \rho, \qquad \rho = 10^{9},
  \label{eq:ratio-bound}
\end{equation}
where $s_{\max} = \max_i \stack{i}$ and $s_{\min} = \min_i \stack{i}$. Each factor of the generating
function in \cref{eq:loo-pgf} changes from $1$ to $z$ when $\stack{j} t$ is of order one.
On the time axis, the largest stack therefore makes its transition first, near $t = 1 /
s_{\max}$, and the smallest stack last, near $t = 1 / s_{\min}$. Because the \ICM{} is scale
invariant, we normalize by the largest stack, with a \emph{headroom} $c$:
\begin{equation}
  \nstack{i} = \frac{\stack{i}}{c\, s_{\max}}, \quad c = 20, \qquad\text{so that}\qquad
  \frac{1}{c\rho} \le \nstack{i} \le \frac{1}{c} \quad\text{and}\quad
  \widetilde{S} := \sum_i \nstack{i} \le \frac{n}{c} .
  \label{eq:normalization}
\end{equation}
The largest stack then transitions near $t = c$ regardless of the instance, and the smallest at most
$\log \rho$ further to the right on the $\log t$ axis.
The headroom $c$ will be explained in the next subsection.

\subsection{The exp-exp rule}
\label{sec:rule}

For an integral over $[0, \infty)$ we substitute
\begin{equation}
  t = \dexp(x) = \exp\bigl( x - e^{-x} \bigr), \qquad
  \int_{0}^{\infty} f(t)\, \mathrm{d}t
  \;=\; \int_{-\infty}^{\infty} f(\dexp(x))\, \dexp'(x)\, \mathrm{d}x,
  \qquad \dexp'(x) = (1 + e^{-x})\, \dexp(x),
  \label{eq:expexp}
\end{equation}
the double-exponential transformation for integrands that decay exponentially at infinity \citep{takahasi1974,mori2001}.
We call this classical rule exp-exp. The map has two branches.
For $x \gg 0$, $t \approx e^{x}$, so the lattice $x \in \step \mathbb{Z}$ is uniform in $\log t$;
for $x \ll 0$, $t \approx \exp(-e^{-x})$ tends to zero double-exponentially.
We apply the trapezoidal rule with step $\step$ on the lattice, truncated to a window $[\xL, \xR)$:
\begin{equation}
  \int_{0}^{\infty} f(t)\, \mathrm{d}t \;\approx\; \sum_{\ell = 0}^{\nodes - 1} w_\ell\, f(t_\ell),
  \qquad t_\ell = \dexp(\xL + \ell \step), \quad w_\ell = \step\, \dexp'(\xL + \ell \step),
  \quad \nodes = (\xR - \xL) / \step .
  \label{eq:rule}
\end{equation}
The nodes and weights depend only on $\step$ and the window, not on the input, so at node $\ell$
the $n$ numbers $\behind{j}(t_\ell) = e^{-\nstack{j} t_\ell}$ are computed once and serve all $n
\times m$ integrands (\cref{sec:allplayers}).

\begin{figure}[b]
  \centering
  \IfFileExists{figures/integrand.pdf}{\includegraphics[width=0.67\linewidth]{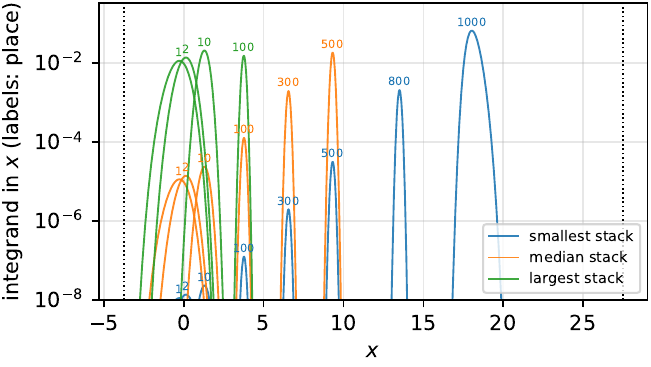}}{\todo{figure pending}}
  \caption{Transformed integrands $\nstack{i} e^{-\nstack{i} t} \loo{i}{r}(t)\, \dexp'(x)$ of
  \cref{eq:representation} in the variable $x$ of the exp-exp map, for the smallest, the median
  and the largest stack of a $1{,}000$-player field with a geometric stack progression of ratio
  $10^{6}$, for a selection of places, under the normalization of \cref{eq:normalization}.
  Labels are finishing places (1 denotes first place). Dotted lines mark the window of \cref{sec:errorbudget}.}
  \label{fig:integrand}
\end{figure}

Two features of the map make the lattice fit these integrands. On the logarithmic branch the
lattice is uniform in $\log t$, and in that variable every integrand is a bump whose shape does
not depend on the scale of the stacks (\cref{fig:integrand}). Approximating the narrow bumps by
Gaussians predicts a trapezoidal-rule error that decays as $e^{-\kappa / (n \step^{2})}$, with a
constant $\kappa$ derived in \cref{sec:errorbudget}.
On the double-exponential branch the map brings the tail $t \to 0$ within a few units of $x$:
the left end of the window can sit at $\xL \approx -3.7$, whereas a lattice uniform in $\log t$
would need $\xL \approx -42$ for the same truncation error (\cref{sec:errorbudget}).
The rules of \MPL{} and \LAQI{} also place their nodes in a logarithmic variable; they differ from
the lattice at $t \to 0$ and in the spacing of the nodes (\cref{sec:ablations}).

\paragraph{Compression and headroom.}
Near the junction of the two branches the lattice is compressed: its local spacing in $\log t$
is $\step\, \mathrm{d} \log t / \mathrm{d}x = \step\, (1 + e^{-x})$, a factor $1.57$ where $t = 1$
($x \approx 0.57$) and $1.05$ where $t = 20$ ($x \approx 3.0$). A bump that lies on the compressed
part is narrower in $x$ than in $\log t$ by that factor and would need a correspondingly finer
step. The headroom of \cref{eq:normalization} moves the transitions away from this part: with
$\nstack{\max} = 1/20$ the largest stack transitions at $t = 20$, and every other transition lies
further to the right, on the logarithmic branch. What remains on the compressed branch are the
integrands for the early places, whose bumps are broad (\cref{fig:integrand}). The price of headroom
is window length, one unit for each factor $e$, and the choice of $c$ balances two costs. The
narrowest bumps lie at $t \approx b^{*} c$ (\cref{lem:width}), where $b^{*} \approx 1.6$ and the
compression is about $1 + 1 / (b^{*} c)$. The number of nodes grows with this factor, through the
step, and with the window length $W \approx 31$; their product is smallest at $b^{*} c \approx W$,
that is, at $c \approx 20$. We take this value for all results of the paper.

\paragraph{The exp-sinh rule.}
The other double-exponential transformation for $[0, \infty)$, $t = \exp(\tfrac{\pi}{2} \sinh x)$,
is the one in common use, designed for integrands that decay
algebraically \citep{takahasi1974,mori2001}. Its resolution varies across the field. Under the
natural normalization by the midrange $\sqrt{s_{\max} s_{\min}}$, the transitions lie
symmetrically about $t = 1$, and the map compresses $\log t$ by $\tfrac{\pi}{2} \cosh x$: a
factor $1.6$ at $t = 1$ and $10.5$ at the extreme stacks of a field of ratio $\rho$. The exp-exp map
has no such gradient on its logarithmic branch, and we use it throughout the paper.

\subsection{Error budget}
\label{sec:errorbudget}

The error of a computed placement probability has three sources: the truncation of the integral to
a finite window (\cref{prop:truncation}), the discretization of the integral by the trapezoidal rule
(\cref{lem:width,eq:kappa}), and rounding in floating-point arithmetic
(\cref{rem:rounding}).

\begin{proposition}[Truncation]
  \label[proposition]{prop:truncation}
  Let $t_{\mathrm{L}} = \dexp(\xL)$ and $t_{\mathrm{R}} = \dexp(\xR)$. Omitting $[0,
  t_{\mathrm{L}})$ from \cref{eq:representation} changes $\Q{i}{1}$ by at most $\nstack{i}
  t_{\mathrm{L}}$, which is a relative error of at most $\widetilde{S}\, t_{\mathrm{L}}$, and
  changes every $\Q{i}{r}$ with $r \ge 2$ by at most $\nstack{i} \widetilde{S}\,
  t_{\mathrm{L}}^{2} / 2$. Omitting $[t_{\mathrm{R}}, \infty)$ changes every $\Q{i}{r}$ by at
  most $e^{-\nstack{i} t_{\mathrm{R}}}$, and every equity of player $i$ by a relative error of
  at most $2\, e^{-\nstack{i} t_{\mathrm{R}}}$ when $e^{-\nstack{i} t_{\mathrm{R}}} \le 1/2$.
\end{proposition}

The proof is in \cref{app:truncation}. On the left only the first place matters, since the
bound for the later places carries an extra factor $t_{\mathrm{L}}$. The first place has the
closed form $\Q{i}{1} = \nstack{i} / \widetilde{S}$ and could be left out of the quadrature,
which would save some $3\%$ of the nodes; for simplicity we integrate it like the other places.

\paragraph{Window.}
By \cref{eq:normalization} the left bound for the first place is at most $(n / c)\,
t_{\mathrm{L}}$ and the right bound at most $e^{-t_{\mathrm{R}} / (c\rho)}$. Requiring both to be
at most a target accuracy $\varepsilon$, with $t \approx \exp(-e^{-x})$ for $x \ll 0$ and $t \approx
e^{x}$ for $x \gg 0$, gives
\begin{equation}
  \xL \;\approx\; -\log \log \bigl( n / (c\, \varepsilon) \bigr)
  \qquad\text{and}\qquad
  \xR \;\approx\; \log \bigl( c \rho \log (1 / \varepsilon) \bigr) .
  \label{eq:window}
\end{equation}
For $\rho = 10^{9}$, $n \le 4{,}000$, $c = 20$ and $\varepsilon = 10^{-16}$, a target below the
rounding floor, this is $\xL \approx -3.7$ and $\xR \approx 27.3$; this paper uses $[\xL, \xR) =
[-3.75, 27.5)$. The window depends on $n$, $\rho$, $c$ and the target $\varepsilon$, but not
on the instance. With $\rho = 10^{4}$ instead, the right end would move to $\xR \approx 15.8$,
which shortens the window, and hence the number of nodes, by more than a third.

\paragraph{Width of the bumps.}
For an interior place $r + 1$, the Poisson-binomial coefficient forms a bump near the time when the
expected number ahead, $\mu_i(t) = \sum_{j \ne i} (1 - \behind{j}(t))$, reaches $r$.
A Gaussian approximation estimates its width in $\log t$ as the standard deviation of the number
ahead divided by the rate at which its mean changes:
\begin{equation}
  \sigma^{2} \;=\; \frac{v}{\mu'^{2}}, \qquad
  v = \sum_{j \ne i} \behind{j} (1 - \behind{j}), \qquad
  \mu' = \frac{\mathrm{d}\mu_i}{\mathrm{d} \log t} = \sum_{j \ne i} a_j \behind{j}, \qquad
  a_j = \nstack{j} t .
  \label{eq:bumpwidth}
\end{equation}

\begin{lemma}[Width estimate]
  \label[lemma]{lem:width}
  Let $g(a) = a^{2} / (2 (e^{a} - 1))$ and $\gmax = \max_{a > 0} g(a) \approx 0.324$, attained
  at $a = b^{*} \approx 1.594$. For every $t > 0$,
  \begin{equation}
    \sigma^{2} \;\ge\; \frac{1}{2 \sum_{j \ne i} g(a_j)} \;\ge\; \frac{1}{2 (n - 1)\, \gmax} ,
    \label{eq:width}
  \end{equation}
  with equality in the first inequality if and only if all $a_j$ are equal, and in the second if
  and only if moreover $a_j = b^{*}$.
\end{lemma}

The proof is in \cref{app:width}. The smallest estimated width is $(2 (n - 1) \gmax)^{-1/2} \approx
1.24\, (n - 1)^{-1/2}$, attained by equal opponent stacks when a fraction $1 - e^{-b^{*}} \approx
80\%$ of the opponents is expected to be ahead. Since $g(a) \le a/2$, we also have $\sigma^2 \ge
1/(\widetilde{S}t)$, consistent with the broad bumps of early places.

\paragraph{Predicted discretization error.}
For a Gaussian of mass $m_0$ and width $\sigma$ in the lattice variable, the trapezoidal
rule has a leading aliasing amplitude of $2 m_0\, e^{-2\pi^{2}
\sigma^{2} / \step^{2}}$ \citep{trefethen2014}. In an equal-stack field every placement
probability is $1 / n$, so the worst absolute error of a placement probability on a lattice of
spacing $\step$ in $\log t$ is predicted, with no fitted constant, as
\begin{equation}
  E_Q(n, \step) \;\approx\; \frac{2}{n} \exp\Bigl( -\frac{\kappa}{(n - 1)\, \step^{2}} \Bigr),
  \qquad \kappa = \frac{\pi^{2}}{\gmax} \approx 30.5 .
  \label{eq:kappa}
\end{equation}
On the exp-exp lattice the spacing at the narrowest bumps, at $t = b^{*} c \approx 32$, is
about $1.03\, \step$, which lowers $\kappa$ to $28.7$; we use $\kappa = 28$ below.
A rigorous bound for the lattice uniform in $\log t$ supports the same exponential rate
(\cref{thm:strip} in \cref{app:strip}). For the exp-exp rule, \cref{eq:kappa} is a prediction, and
the measured worst errors are within a factor \PredRatioMin{} to \PredRatioMax{} of it (\cref{app:error}).

\begin{table}[t]
  \centering
  \small
  \caption{Measured worst relative error of equities over the extremal set of \cref{sec:accuracy},
  by step and field size, on the window $[-3.75, 27.5)$. Shaded cells are at the rounding floor
  (below $5 \times 10^{-14}$). The corresponding table for placement probabilities is
  \cref{tab:sweeph} in \cref{app:error}.}
  \label{tab:step}
  \IfFileExists{tables/sweep_h_relV.tex}{\input{tables/sweep_h_relV}}{\todo{table pending}}
\end{table}

\paragraph{Choosing the step.}
Solving \cref{eq:kappa} for the step at which the predicted error meets a target $\varepsilon$
gives
\begin{equation}
  \step \;\approx\; \sqrt{\frac{\kappa}{(n - 1) \log \bigl( 2 / (n \varepsilon) \bigr)}},
  \label{eq:step}
\end{equation}
so the step shrinks as $n^{-1/2}$, up to the slowly varying logarithm, and the number of nodes
grows as $n^{1/2}$. With $\kappa = 28$ and $\varepsilon = 10^{-15}$, the level of the rounding
floor, \cref{eq:step} gives $\step \approx 1/10$, $1/32$ and $1/62$ for $n = 100$, $1{,}000$
and $4{,}000$. \Cref{tab:step} shows the measured worst relative error of an equity over the
extremal set of \cref{sec:accuracy} as the step is refined: the error falls as \cref{eq:kappa}
predicts until it meets the rounding floor. For each field size in \cref{tab:step}, the experiments
of \cref{sec:results} use the coarsest power-of-two step that reaches this floor, from $1/16$ at
$n=100$ to $1/64$ at $n=4{,}000$. For smaller fields, aliasing from the exponential tails of the
densities motivates an empirically chosen cap of $1/8$ (\cref{app:error}).
\Cref{app:constants} specifies the step for all field sizes up to $4{,}000$.

\begin{remark}[Rounding]
  \label[remark]{rem:rounding}
  The coefficients at each node are products and sums of non-negative numbers, and
  \cref{sec:allplayers} delivers them with an absolute error of $O(n \ueps)$; the quadrature
  sum over $\nodes$ positive weights adds $O(\nodes \ueps)$. The resulting floor is
  $O((n + \nodes)\ueps)$ in the worst case and $O(\sqrt{n + \nodes}\, \ueps)$ typically, about
  $10^{-14}$ for $n \approx \nodes \approx 10^{3}$. This accounting is heuristic, and we treat
  the floor as an estimate to be confirmed by measurement (\cref{sec:accuracy}).
\end{remark}

%% file: tables/sweep_h_relV.tex
\begin{tabular}{@{}rrcccccc@{}}
\toprule
$1/\step$ & $\nodes$ & $n=100$ & $n=200$ & $n=500$ & $n=1{,}000$ & $n=2{,}000$ & $n=4{,}000$ \\
\midrule
8 & 250 & $3.7 \times 10^{-10}$ & $4.2 \times 10^{-6}$ & $6.4 \times 10^{-4}$ & $6.0 \times 10^{-3}$ & -- & -- \\
12 & 375 & \cellcolor{floorshade}$1.4 \times 10^{-15}$ & $1.9 \times 10^{-11}$ & $4.9 \times 10^{-6}$ & $4.0 \times 10^{-4}$ & $2.8 \times 10^{-3}$ & -- \\
16 & 500 & \cellcolor{floorshade}$1.1 \times 10^{-15}$ & \cellcolor{floorshade}$2.2 \times 10^{-15}$ & $8.9 \times 10^{-9}$ & $7.1 \times 10^{-6}$ & $4.9 \times 10^{-4}$ & $2.2 \times 10^{-3}$ \\
24 & 750 & \cellcolor{floorshade}$1.3 \times 10^{-15}$ & \cellcolor{floorshade}$2.9 \times 10^{-15}$ & \cellcolor{floorshade}$3.7 \times 10^{-15}$ & $8.7 \times 10^{-10}$ & $3.0 \times 10^{-6}$ & $1.9 \times 10^{-4}$ \\
32 & 1{,}000 & \cellcolor{floorshade}$1.3 \times 10^{-15}$ & \cellcolor{floorshade}$2.8 \times 10^{-15}$ & \cellcolor{floorshade}$4.2 \times 10^{-15}$ & \cellcolor{floorshade}$1.1 \times 10^{-14}$ & $3.4 \times 10^{-9}$ & $6.5 \times 10^{-6}$ \\
48 & 1{,}500 & \cellcolor{floorshade}$1.2 \times 10^{-15}$ & \cellcolor{floorshade}$2.8 \times 10^{-15}$ & \cellcolor{floorshade}$4.5 \times 10^{-15}$ & \cellcolor{floorshade}$1.1 \times 10^{-14}$ & \cellcolor{floorshade}$2.4 \times 10^{-14}$ & $3.8 \times 10^{-10}$ \\
64 & 2{,}000 & \cellcolor{floorshade}$1.7 \times 10^{-15}$ & \cellcolor{floorshade}$2.3 \times 10^{-15}$ & \cellcolor{floorshade}$2.7 \times 10^{-15}$ & \cellcolor{floorshade}$1.5 \times 10^{-14}$ & \cellcolor{floorshade}$3.2 \times 10^{-14}$ & \cellcolor{floorshade}$4.0 \times 10^{-14}$ \\
\bottomrule
\end{tabular}

%% file: sections/04-allplayers.tex
\section{All players at once: stable leave-one-out deconvolution}
\label{sec:allplayers}

At a fixed node $t$, the integrals of \cref{eq:representation} need the leave-one-out coefficient
$\loo{i}{r}(t)$ of \cref{eq:loo-pgf} for every player $i$ and every place index $r < m$. This
section obtains all of them from one product in $O(n^{2})$ operations, with a stability
guarantee. Throughout, $t$ is fixed and suppressed, $\behind{j} = e^{-\nstack{j} t}$ is the
probability that player $j$ finishes behind the conditioned player, and the stacks are
normalized and sorted so that $\nstack{1} \le \dots \le \nstack{n}$, hence $\behind{1} \ge
\dots \ge \behind{n}$.

\subsection{The full product}
\label{sec:fullproduct}

The full generating function $\pgf(z) = \prod_{j=1}^{n} (\behind{j} + (1-\behind{j}) z) =
\sum_{r=0}^{n} \coef{r} z^{r}$ is computed by $n$ successive convolutions with the two-term
vectors $[\behind{j}, 1 - \behind{j}]$, in $\sum_{j} O(j) = O(n^{2})$ operations. The two entries
of each vector are evaluated separately, as $e^{-\nstack{j} t}$ and
$-\operatorname{expm1}(-\nstack{j} t)$, since the difference $1 - e^{-x}$ formed by subtraction
loses all relative accuracy for small $x$. Each convolution step acts on non-negative terms and
does not amplify the error already present, so the computed coefficients satisfy
\begin{equation}
  \sum_{r} \bigl| \widehat{\coef{r}} - \coef{r} \bigr| \;\le\; \theta, \qquad
  \theta \le (4 + e^{-2})\, n \ueps + O(\ueps^{2})
  \label{eq:product-error}
\end{equation}
(\cref{app:product}). Factors and coefficients that underflow are replaced by zero, an absolute
error below $2^{-1022}$ each, which we ignore. The coefficients are probabilities and sum to
one, a fact the error analysis below uses repeatedly. Pairing the factors in a tree and
multiplying the pairs with the fast Fourier transform would form the product in $O(n \log^{2} n)$
operations; however, the rest of the work at a node is $O(n^{2})$ in any case, and direct
convolution is the most accurate of the standard methods for this distribution
\citep{hong2013,biscarri2018}.

\subsection{Stable division by a linear factor}
\label{sec:deconvolution}
\label{sec:clamp}

Since $\loopgf{i}(z) = \pgf(z) / (\behind{i} + (1-\behind{i}) z)$, the leave-one-out
coefficients satisfy $\coef{r} = \behind{i} \loo{i}{r} + (1 - \behind{i}) \loo{i}{r-1}$ for all
$r$, with $\loo{i}{-1} = \loo{i}{n} = 0$. We work with $\scaledloo{i}{r} = \behind{i}
\loo{i}{r}$, which is exactly what \cref{eq:representation} needs once $\nstack{i}
e^{-\nstack{i} t}$ is written as $\nstack{i} \behind{i}$. With the multiplier $\mult{i} = (1 -
\behind{i}) / \behind{i}$ the relation reads $\coef{r} = \scaledloo{i}{r} + \mult{i}
\scaledloo{i}{r-1}$, and it can be solved in two directions:
\begin{align}
  \text{forward:}\quad & \scaledloo{i}{r} = \coef{r} - \mult{i}\, \scaledloo{i}{r-1},
    \qquad r = 0, 1, \dots, \quad \scaledloo{i}{-1} = 0,
    \label{eq:forward} \\
  \text{backward:}\quad & \scaledloo{i}{r-1} = \bigl( \coef{r} - \scaledloo{i}{r} \bigr) / \mult{i},
    \qquad r = n, n-1, \dots, \quad \scaledloo{i}{n} = 0 .
    \label{eq:backward}
\end{align}
Either recurrence costs $O(1)$ per coefficient. Both are synthetic division of $\pgf$ by a
linear factor, that is, polynomial deflation, which is unstable in general;
\citet{knyazev2026} mention \cref{eq:forward} but set it aside for that reason.
The classical analysis of \citet{peters1971} shows that deflation is stable if it starts from the
constant term when the removed root is large in modulus, and from the leading coefficient when it is
small. The roots of $\pgf$ are the negative reals $-1/\mult{j}$, so here the two regimes are
separated by a single threshold, $|{-1/\mult{i}}| \ge 1$ if and only if $\behind{i} \ge 1/2$, and
the direction alone controls the error.

\begin{theorem}[Stability of direction-switched deconvolution]
  \label{thm:stability}
  Let the computed full-product coefficients satisfy $\sum_r |\widehat{\coef{r}} - \coef{r}|
  \le \theta$, let $\behind{i}$ and $1 - \behind{i}$ be evaluated as in
  \cref{sec:fullproduct}, and let \cref{eq:forward} be used when $\behind{i} \ge 1/2$ and
  \cref{eq:backward} otherwise, in floating-point arithmetic with unit roundoff $\ueps$. Then
  for every $r$,
  \begin{equation}
    \bigl| \widehat{\scaledloo{i}{r}} - \scaledloo{i}{r} \bigr|
    \;\le\; \theta + 7\ueps \sum_{s} \coef{s} + O(\ueps^{2})
    \;=\; \theta + 7\ueps + O(\ueps^{2}) .
    \label{eq:stability-bound}
  \end{equation}
  With $\theta$ from \cref{eq:product-error}, every scaled leave-one-out coefficient has absolute
  error at most $(4 + e^{-2})\, n \ueps + 7 \ueps + O(\ueps^{2}) < (4.2\, n + 7)\ueps$,
  uniformly in $r$. If the direction is
  chosen against the rule, the error committed at step $s$ is amplified by
  $\max(\mult{i}, 1/\mult{i})^{\,|r - s|}$ at step $r$.
\end{theorem}

\begin{proof}[Proof idea]
  A perturbation $\delta$ of the carried value changes the next value by $-\mult{i} \delta$ in
  \cref{eq:forward} and by $-\delta / \mult{i}$ in \cref{eq:backward}, a factor of modulus at
  most one under the rule, so the error at any step is at most the sum of the local errors of
  the steps traversed. Each local error is a small multiple of $\ueps \coef{r}$ plus the error
  of the input coefficient, and the sum over $r$ is bounded because $\sum_r \coef{r} = 1$ and
  the input errors sum to $\theta$. \Cref{app:stability} gives the accounting.
\end{proof}

The bound is absolute, in units of total probability, so small coefficients are accurate in
absolute but not in relative terms. The prefix--suffix scheme of \citet{knyazev2026}, which
forms every leave-one-out coefficient as a sum of non-negative products, keeps relative accuracy
at the price of $O(nm^{2})$ per node; \cref{sec:loo} measures both.

\paragraph{Clamping.}
The recurrences preserve a simple invariant, which we enforce in floating point.

\begin{lemma}[Invariant]
  \label[lemma]{lem:invariant}
  For all $i$ and $r$, $0 \le \scaledloo{i}{r} \le \coef{r}$.
\end{lemma}

\begin{proof}
  $\coef{r} = \scaledloo{i}{r} + (1 - \behind{i}) \loo{i}{r-1}$ with both terms non-negative.
\end{proof}

After each step of \cref{eq:forward,eq:backward} we clamp the carried value to $[0, \widehat{\coef{r}}]$,
with the index shifted by one in the backward direction. In exact arithmetic this is the identity.
In floating point it retains an $O(n\ueps)$ absolute error bound (\cref{app:stability}) while
removing a residual that this bound alone does not exclude. At a node where the smallest stack's
$\behind{i}$ has just dropped below $1/2$, a residue of order $\ueps$ can arise where the true value
is zero. Amplified by the quadrature weight, it produces a relative error of up to $n \rho\, \ueps
\step$ in the first-place probability, which for $\rho = 10^{9}$ exceeds every other term of the
budget (\cref{app:fartail}). The clamp removes the residue at its source, because the bound
$\widehat{\coef{n-2}}$ is itself zero where the true value vanishes.

\subsection{Algorithm and cost}
\label{sec:algorithm}
\label{sec:implementation}

\begin{algorithm}[b]
  \caption{\DEICM{}: probabilities of the first $m$ places and equities for all players.}
  \label{alg:deicm}
  \begin{algorithmic}[1]
    \Require stacks $\stack{1..n}$, prizes $\prize{1..m}$, grid $(\step, \xL, \xR)$
    \State normalize by $c\,s_{\max}$ ($c = 20$) and sort: $\nstack{1} \le \dots \le \nstack{n}$
    \State $\widehat{Q} \gets 0_{n \times m}$
    \For{each node $x = \xL, \xL + \step, \dots < \xR$}
      \State $t \gets \dexp(x)$;\quad $w \gets \step\, \dexp'(x)$
      \State $\behind{j} \gets \exp(-\nstack{j} t)$;\quad $p_j \gets -\operatorname{expm1}(-\nstack{j} t)$ for all $j$ \Comment{$\behind{j}$ non-increasing in $j$}
      \State compute $\coef{0..n}$ by $n$ convolutions with $[\behind{j}, p_j]$ \Comment{$O(n^{2})$}
      \State $b \gets |\{ j : \behind{j} \ge 1/2 \}|$ \Comment{players $1..b$ go forward}
      \State $\scaledloo{i}{-1} \gets 0$ for $i \le b$;\quad $\scaledloo{i}{n} \gets 0$ for $i > b$
      \For{$r = 0, \dots, m-1$} \Comment{forward group, vectorized over $i \le b$}
        \State $\scaledloo{i}{r} \gets \operatorname{clamp}\bigl(\coef{r} - \tfrac{p_i}{\behind{i}}\, \scaledloo{i}{r-1},\, 0,\, \coef{r}\bigr)$;\quad $\widehat{Q}_{i, r+1} \mathrel{+}= w\, \scaledloo{i}{r}$
      \EndFor
      \For{$r = n, \dots, 1$} \Comment{backward group, vectorized over $i > b$}
        \State $\scaledloo{i}{r-1} \gets \operatorname{clamp}\bigl(\tfrac{\behind{i}}{p_i} (\coef{r} - \scaledloo{i}{r}),\, 0,\, \coef{r-1}\bigr)$;\quad \textbf{if} $r \le m$: $\widehat{Q}_{i, r} \mathrel{+}= w\, \scaledloo{i}{r-1}$
      \EndFor
    \EndFor
    \State $\widehat{Q}_{i, r} \gets \nstack{i}\, \widehat{Q}_{i, r}$ for all $i, r$; undo the sort
    \State \Return $\widehat{Q}$ and $\widehat{V}_i = \sum_r \prize{r} \widehat{Q}_{i,r}$
  \end{algorithmic}
\end{algorithm}

\Cref{alg:deicm} puts these steps together; the probability of finishing out of the money, when $m <
n$, is $1 - \sum_{r \le m} \widehat{Q}_{i,r}$. It computes the $n \times m$ matrix in
$\Theta\bigl(\nodes\,(n^{2} + b m + (n - b) n)\bigr) = O(\nodes n^{2})$ time and $O(nm)$ memory,
where $b$ is the size of the forward group at a node. The forward recurrence can stop at place $m$,
whereas the backward one must start from place $n$. The full product alone costs $O(n^{2})$, so the
leave-one-out step adds no asymptotic cost.
Cached prefix and suffix products \citep{knyazev2026} give the matrix in $O(nm^{2})$ per node, which
has a lower asymptotic cost when $m \lesssim \sqrt{n}$, and the adjoint pass of \citet{wolters2026}
gives all equities, but not the individual places, in $O(nm)$.
Both methods avoid subtractions; the deconvolution accepts them, and this is what lowers its cost
for the placement matrix with $m \gtrsim \sqrt{n}$ (\cref{tab:comparison} in \cref{sec:related-pl}).

\paragraph{Implementation.}
The inner loops of \cref{alg:deicm} run over players for a fixed place, so the placement matrix
is stored place-major, and the players of a direction group, which are contiguous after sorting,
are processed in blocks. The recurrences of different players are independent, so the loops
vectorize. The $2n$ evaluations of $\exp$ and $\operatorname{expm1}$ are the only transcendental
operations at a node. Each node adds its contribution to all $nm$ entries of the matrix, so the
nodes are processed in blocks of $32$, which streams the matrix through memory once per block
instead of once per node. The nodes are independent, so the outer loop parallelizes across
threads; a fixed assignment of nodes to threads and a fixed reduction order keep the result
bit-for-bit reproducible.

%% file: sections/05-results.tex
\section{Numerical results}
\label{sec:results}

All results come from a stand-alone implementation of \cref{alg:deicm} written for this paper,
run on the grid of \cref{sec:errorbudget} unless a sweep varies it, on an AMD Ryzen~9 9950X3D
(16 cores) under WSL2 with GCC~14.2 at \texttt{-O3 -march=native}.

\subsection{Validation suite and accuracy}
\label{sec:references}
\label{sec:accuracy}

\begin{table}[b]
  \centering
  \caption{Worst errors over the validation suite by field size $n$, at the step and node count of
  \cref{sec:errorbudget} (first two rows); the third row counts the instance--reference comparisons.
  Cumulative probabilities sum places $1$ through $r$. Row sums (the places of one player) and
  column sums (the players at one place) show deviations from one.
  Probability errors use the available exact references; row sums are checked for $m=n$ and column
  sums for every computed place.}
  \label{tab:accuracy}
  \IfFileExists{tables/accuracy_by_n.tex}{\input{tables/accuracy_by_n}}{\todo{table pending}}
\end{table}

Claims at the level of $10^{-14}$ cannot be checked by Monte Carlo, whose standard error decreases
as the inverse square root of the number of samples, nor by general-purpose exact evaluation, which
is infeasible beyond about $30$ players. The suite therefore consists of instances whose placement
probabilities or equities have closed forms or low-dimensional exact dynamic programs that scale to
$n = 4{,}000$ (\cref{app:references}): winner-take-all, two-prize, flat and linear prize structures
on arbitrary stacks, and arbitrary prize structures on equal stacks, on two stack sizes, on equal
stacks with up to ten exceptions, and on small fields with $n \le 20$. The stacks range from
equal to geometric progressions, Zipf and log-normal draws and single outliers, with ratios up
to $\rho = 10^{9}$; the prizes range from winner-take-all to flat, linear, top-heavy and realistic
payouts and to satellites paying the top half, the top $80\%$ or all but the last player
(\cref{app:suite}). The references are evaluated in double-double arithmetic, so their own
rounding error is far below the errors measured. The suite has \NumInstances{} instances with
$n$ from $3$ to \MaxN{}, formed from \NumStacks{} stack vectors and \NumPrizeStructures{} prize
structures; \NumCrossChecked{} of them have multiple references and are checked against
each, \NumReferencePairs{} comparisons in all.

\Cref{tab:accuracy} gives the worst errors by field size. Every field size is at the rounding
floor of \cref{rem:rounding}: over all instances the worst relative error of an equity is
\WorstRelV{} and the worst absolute error of a placement probability \WorstAbsQ{}, with $250$
nodes up to $30$ players and $2{,}000$ up to $4{,}000$. The row and column sums of the matrix,
accumulated in double-double arithmetic, deviate from one by at most \WorstSumDev{} and
\WorstColDev{}; the column sums grow faster with $n$ because in fields with large groups of equal
stacks thousands of identical rounding errors add coherently. For the families with exact full matrices we also
compare the cumulative probabilities $F_{i,r} = \sum_{r' \le r} \Q{i}{r'}$, whose worst relative
error is \WorstRelF{}. Every equity under a non-increasing prize structure is a non-negative
combination of them, $\V{i} = \sum_{r} (\prize{r} - \prize{r+1}) F_{i,r}$, so this figure bounds
the relative error of the equities implied by the computed matrix for every prize structure at
once. Probabilities that are themselves small are accurate in absolute but not in relative
terms, since the bound of the deconvolution is absolute (\cref{sec:deconvolution}): the worst
relative error of a probability above $10^{-10}$ is \WorstRelQten{}.

The sweeps and comparisons below run on a smaller \emph{extremal set} of $25$ instances per field
size, chosen from the error budget: groups of equal stacks, which minimize the estimated width
(\cref{lem:width}), with the prize structures that expose the deep places, and fields spanning the
full ratio $\rho$, for which the truncation bounds are tight.
To check that it captures the largest errors of the validation suite, we run the whole suite at
twice the step of \cref{sec:errorbudget}, where the worst errors are above the floor: they agree
with those of the extremal set exactly for probabilities and to within $0.2\%$ for equities (\cref{app:suite}).

\subsection{Quadrature rules}
\label{sec:ablations}

\begin{figure}[b]
  \centering
  \IfFileExists{figures/sweep_rule.pdf}{\includegraphics[width=0.92\linewidth]{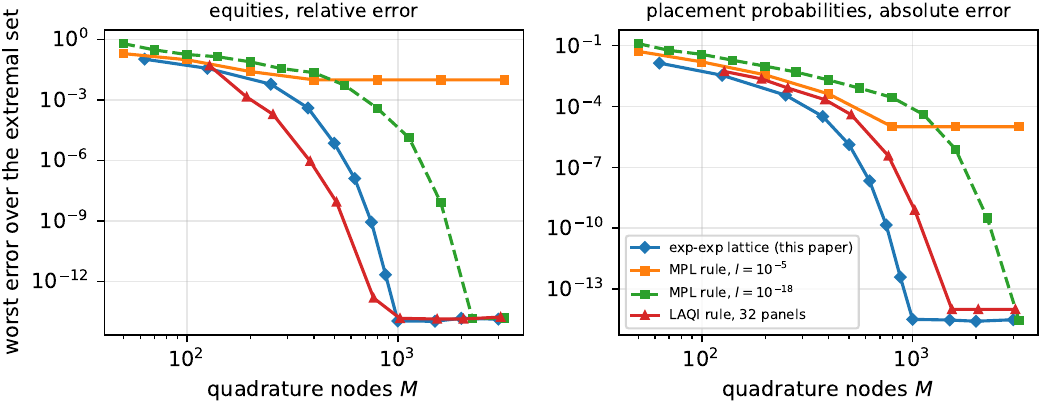}}{\todo{figure pending}}
  \caption{Worst relative error of equities (left) and worst absolute error of placement
  probabilities (right) over the extremal set with $n \le 1{,}000$, against the number of
  nodes, for the exp-exp lattice and for the quadrature rules of \MPL{} and \LAQI{} in the same
  engine. \MPL{} rule: published margin $l = 10^{-5}$ (solid) and $l = 10^{-18}$ (dashed).
  Only the quadrature rule is exchanged: the curves compare the rules as node layouts, not the
  \MPL{} and \LAQI{} methods as published.}
  \label{fig:rule}
\end{figure}

\Cref{fig:rule} compares the exp-exp lattice with the quadrature rules of \MPL{} and \LAQI{}, in the
same leave-one-out engine and on the extremal set with $n \le 1{,}000$. The \MPL{} rule is
Gauss--Legendre in the Gumbel variable on a window set from the range of the scores with a quantile
margin $l$; it is shown with its published margin $l = 10^{-5}$ and with $l = 10^{-18}$, which makes
the omitted tails negligible at the target accuracy. The \LAQI{} rule is composite Gauss--Legendre
in $\log(1 + t)$ on $32$ panels with the breaks set from the instance. The engine integrates every
place, including the first, and applies no output rescaling; in the published implementations,
\MPL{} evaluates the first place analytically and \LAQI{} rescales its equities to the prize pool.

At its published margin the \MPL{} rule levels off at an absolute error of $10^{-5}$ in the probabilities
and a relative error of order $n\, l$ in the equities. Once truncation is negligible, the rules differ
only in how their nodes are spread over $\log t$. The exp-exp lattice reaches the floor at $1{,}000$
nodes for both quantities. Before the floor, the \LAQI{} rule is more accurate than the lattice for
equities and less accurate for probabilities; it reaches the floor at $1{,}024$ nodes for equities
and at $1{,}536$ for probabilities. With its breaks fixed for the stack-ratio bound instead of set
from the instance, it needs $1{,}536$ nodes for both quantities in our implementation (not shown).
The \MPL{} rule with the wide margin needs $2{,}260$ and $3{,}200$.
The per-node cost of \cref{alg:deicm} does not depend on the rule, so at a given accuracy the order
of the curves is also the order of the running times.

\subsection{Leave-one-out schemes}
\label{sec:loo}

\begin{table}[t]
  \centering
  \small
  \caption{Worst errors over the extremal set with $n \le 1{,}000$ when the leave-one-out step of
  \cref{alg:deicm} is replaced, on the grid of \cref{sec:errorbudget}.
  The prefix--suffix row covers only instances with $n m^{2} \le 2 \times 10^{8}$.
  ``Overflow'' indicates a non-finite result on at least one instance.}
  \label{tab:loo}
  \IfFileExists{tables/ablation_loo.tex}{\input{tables/ablation_loo}}{\todo{table pending}}
\end{table}

\Cref{tab:loo} runs \cref{alg:deicm} with its leave-one-out step replaced: each single-direction
recurrence alone, the direction rule of \cref{sec:deconvolution} with and without the clamp, and
the prefix--suffix products of \citet{knyazev2026}. Neither single direction is usable: each
overflows on some instances, and even with the clamp it loses every digit on others. The
direction rule alone brings the equity error down to the level of the residue described in
\cref{sec:clamp}, and the clamp removes it. The prefix--suffix products agree with the clamped
deconvolution to the rounding floor on equities and on absolute probability errors, and they keep
relative accuracy on small probabilities; on the instances their row covers, they take about $27$
times the total running time.

\subsection{Runtime}
\label{sec:runtime}

\begin{table}[t]
  \centering
  \small
  \caption{Median wall-clock time over three runs of \cref{alg:deicm} on the grid of \cref{sec:errorbudget},
  with the step and node count of each field size, on one and eight threads, for the full matrix and
  for $m = 0.15n$ paid places. Stacks are log-normal draws.}
  \label{tab:runtime}
  \IfFileExists{tables/runtime.tex}{\input{tables/runtime}}{\todo{table pending}}
\end{table}

\Cref{tab:runtime} reports the wall-clock time of \cref{alg:deicm} on one thread and on eight.
The single-thread time grows as $n^{2}$ per node, as the cost analysis of \cref{sec:algorithm}
predicts, and the number of nodes grows as $n^{1/2}$, in powers of two: the step halves between
$n = 1{,}000$ and $n = 2{,}000$, which adds a factor of two, while $n = 2{,}000$ and $n = 4{,}000$
share a step. Eight threads give a speed-up of \SpeedupEightMin{} to \SpeedupEightMax{} for
$n \ge 1{,}000$. For comparison, our Monte Carlo implementation at $n = 4{,}000$ takes
\McPerSampleFourK{} microseconds per sampled finishing order. Estimating a probability of $10^{-3}$
with a relative standard error of $10^{-3}$ would require about $10^9$ samples, resulting in about
two CPU-days. \MPL{} reports about $2$\,ms per query for $K = 10$ places with $200$ nodes on a GPU
in single precision \citep{knyazev2026}, and \LAQI{} $2.47$\,s for the equities of $4{,}000$ players
with $192$ nodes in Node.js \citep{wolters2026}; implementations, outputs and accuracies differ, so
these figures are not directly comparable with ours.

%% file: tables/accuracy_by_n.tex
{\footnotesize\setlength{\tabcolsep}{2pt}
\begin{tabular}{@{}lccccccccc@{}}
\toprule
$n$ & $3$ & $10$ & $30$ & $100$ & $200$ & $500$ & $1{,}000$ & $2{,}000$ & $4{,}000$ \\
\midrule
step $\step$ & $1/8$ & $1/8$ & $1/8$ & $1/16$ & $1/16$ & $1/32$ & $1/32$ & $1/64$ & $1/64$ \\
nodes $\nodes$ & $250$ & $250$ & $250$ & $500$ & $500$ & $1{,}000$ & $1{,}000$ & $2{,}000$ & $2{,}000$ \\
comparisons & $152$ & $250$ & $180$ & $180$ & $180$ & $180$ & $180$ & $180$ & $180$ \\
\midrule
equities (rel.) & $4.3\mathrm{e}{-16}$ & $5.8\mathrm{e}{-16}$ & $8.3\mathrm{e}{-16}$ & $1.8\mathrm{e}{-15}$ & $3.1\mathrm{e}{-15}$ & $4.8\mathrm{e}{-15}$ & $1.1\mathrm{e}{-14}$ & $3.2\mathrm{e}{-14}$ & $4.0\mathrm{e}{-14}$ \\
probabilities (abs.) & $1.7\mathrm{e}{-16}$ & $3.3\mathrm{e}{-16}$ & $4.4\mathrm{e}{-16}$ & $4.4\mathrm{e}{-16}$ & $7.8\mathrm{e}{-16}$ & $1.9\mathrm{e}{-15}$ & $3.2\mathrm{e}{-15}$ & $5.9\mathrm{e}{-15}$ & $1.1\mathrm{e}{-14}$ \\
cumulative probabilities (rel.) & $3.1\mathrm{e}{-16}$ & $5.8\mathrm{e}{-16}$ & $8.3\mathrm{e}{-16}$ & $1.6\mathrm{e}{-15}$ & $3.2\mathrm{e}{-15}$ & $1.2\mathrm{e}{-14}$ & $2.3\mathrm{e}{-14}$ & $3.0\mathrm{e}{-14}$ & $6.0\mathrm{e}{-14}$ \\
row sums & $2.5\mathrm{e}{-16}$ & $3.7\mathrm{e}{-16}$ & $4.2\mathrm{e}{-16}$ & $1.0\mathrm{e}{-15}$ & $2.1\mathrm{e}{-15}$ & $3.0\mathrm{e}{-15}$ & $5.9\mathrm{e}{-15}$ & $1.1\mathrm{e}{-14}$ & $2.2\mathrm{e}{-14}$ \\
column sums & $3.3\mathrm{e}{-16}$ & $1.2\mathrm{e}{-15}$ & $2.6\mathrm{e}{-15}$ & $7.0\mathrm{e}{-15}$ & $1.1\mathrm{e}{-14}$ & $1.8\mathrm{e}{-14}$ & $4.4\mathrm{e}{-13}$ & $9.0\mathrm{e}{-14}$ & $4.7\mathrm{e}{-13}$ \\
\bottomrule
\end{tabular}}

%% file: tables/ablation_loo.tex
{\footnotesize
\begin{tabular}{@{}lccc@{}}
\toprule
leave-one-out & equities (rel.) & probabilities (abs.) & probabilities $\ge 10^{-10}$ (rel.) \\
\midrule
forward only & overflow & overflow & overflow \\
forward only, clamped & $3.6 \times 10^{2}$ & $2.4 \times 10^{2}$ & $2.7 \times 10^{11}$ \\
backward only & overflow & overflow & overflow \\
backward only, clamped & $6.3 \times 10^{-1}$ & $9.5 \times 10^{-3}$ & $1.0 \times 10^{0}$ \\
direction rule & $9.6 \times 10^{-8}$ & $3.2 \times 10^{-15}$ & $4.6 \times 10^{-8}$ \\
direction rule, clamped & $1.1 \times 10^{-14}$ & $3.2 \times 10^{-15}$ & $1.9 \times 10^{-8}$ \\
prefix--suffix products & $7.3 \times 10^{-15}$ & $2.6 \times 10^{-15}$ & $3.6 \times 10^{-14}$ \\
\bottomrule
\end{tabular}}

%% file: tables/runtime.tex
{\setlength{\tabcolsep}{3.5pt}\begin{tabular}{@{}rrrrrrr@{}}
\toprule
 & & & \multicolumn{2}{c}{full matrix} & \multicolumn{2}{c@{}}{$m = 0.15n$} \\
\cmidrule(lr){4-5}\cmidrule(l){6-7}
$n$ & $1/\step$ & $\nodes$ & 1 thread & 8 threads & 1 thread & 8 threads \\
\midrule
100 & 16 & 500 & 1.8\,ms & 0.8\,ms & 1.2\,ms & 0.8\,ms \\
1{,}000 & 32 & 1{,}000 & 200.0\,ms & 36.8\,ms & 104.7\,ms & 17.3\,ms \\
2{,}000 & 64 & 2{,}000 & 1.83\,s & 316.4\,ms & 792.7\,ms & 132.7\,ms \\
4{,}000 & 64 & 2{,}000 & 7.58\,s & 1.58\,s & 3.05\,s & 505.8\,ms \\
\bottomrule
\end{tabular}}

%% file: sections/06-related.tex
\section{Related work}
\label{sec:related}

\subsection{Self-organizing lists and caching}
\label{sec:related-mtf}

The move-to-front (MTF) rule maintains a list of $n$ records; a requested record, drawn with
probability $p_i$ independently of the past, is moved to the front. The stationary distribution
of the resulting chain on permutations is the Plackett--Luce distribution \citep{hendricks1972},
so the stationary MTF list is an \ICM{} finishing order with $p_i = \stack{i} / \stacksum$, and
the depth of a record, its \emph{search cost}, is its place minus one. Two performance measures
of self-organizing search are special cases of the \ICM{}. With linear prizes, the equity is an
affine function of the expected search cost, whose average over requests \citet{mccabe1965} gave
in closed form. With unit prizes for the top $C$ places, it is the item's hit probability of a
least-recently-used (LRU) cache of capacity $C$.

The distribution of the search cost was first derived by \citet{burville1973}, as an alternating
sum over subsets of records. The integral representation appeared in two steps.
\citet{flajolet1992} obtained, by formal-language methods, an integral form for the search-cost
distribution and the LRU fault probability, and proposed evaluating it by numerical integration.
\citet{fill1996holst} rederived it by embedding the chain in Poisson time, where the search
cost of a record, conditioned on its exponential age, is Poisson-binomial. Their Theorem~1
divides the full product by one record's factor inside the integral, which is the leave-one-out
operation of \cref{sec:deconvolution} at the level of generating functions;
\cref{thm:representation} is its per-record statement.
The representation has since served exact and asymptotic analyses of MTF and LRU
\citep{fill1996limits,jelenkovic1999,long2026}, always for aggregate quantities;
none of the works cited here computes the placement distributions of all records to a
controlled accuracy.

\subsection{Poisson-binomial computation and quadrature}
\label{sec:related-pb}

At a fixed node the coefficients $\coef{r}$ form a Poisson-binomial distribution;
\citet{hong2013} and \citet{biscarri2018} survey the methods for computing it. Leave-one-out
quantities arise in conditional Bernoulli sampling, where \citet{chen1997} compute them for all
elements through power sums. The direction rule of \cref{thm:stability} is classical in
polynomial deflation \citep{peters1971} but, to our knowledge, has not been recorded for this
distribution.
Double-exponential transformations \citep{takahasi1974,mori2001,trefethen2014} are usually
applied adaptively to one integral at a time; here the step and the window are fixed once for a
class of inputs, and \cref{sec:ablations} compares the resulting lattice with Gauss rules as a
node layout.

\subsection{Plackett--Luce rank marginals and the \ICM{} in practice}
\label{sec:related-pl}

\begin{table}[t]
  \centering
  \footnotesize
  \setlength{\tabcolsep}{4pt}
  \renewcommand{\arraystretch}{1.1}
  \caption{Design choices of the three numerical implementations of the representation;
  \cref{sec:related-pl} discusses \MPL{} and \LAQI{}. $K$ is the number of places returned by \MPL{}.}
  \label{tab:comparison}
  \begin{tabular}{@{}>{\raggedright\arraybackslash}p{2.0cm}>{\raggedright\arraybackslash}p{4.25cm}>{\raggedright\arraybackslash}p{4.25cm}>{\raggedright\arraybackslash}p{4.25cm}@{}}
    \toprule
    & \DEICM{} (this paper) & \MPL{} \citep{knyazev2026} & \LAQI{} \citep{wolters2026} \\
    \midrule
    output & $n \times m$ matrix and equities & top-$K$ matrix; first place analytic & equities, rescaled to the prize pool \\
    cost per node & $O(n^{2})$ & $O(nK^{2})$ & $O(nm)$ \\
    \midrule
    variable & $x$ with $t = \exp(x - e^{-x})$ & Gumbel, $-\log t$ & $u = \log(1 + t)$ \\
    rule & trapezoidal on a lattice & Gauss--Legendre & composite Gauss--Legendre \\
    nodes & fixed for the admissible inputs: set by $n$, $\rho$ and the target accuracy $\varepsilon$ & per instance: a window from the range of the scores with quantile margin $l$ & per instance: upper limit from the smallest stack and a tail tolerance; $32$ panels \\
    leave-one-out & direction-switched deconvolution with clamp & cached prefix and suffix products & prefix states and adjoint pass \\
    arithmetic & subtractions; absolute bound (\cref{thm:stability}) & sums of non-negative terms & sums of non-negative terms \\
    \bottomrule
  \end{tabular}
\end{table}

Placement probabilities of Plackett--Luce policies serve as propensities in off-policy evaluation
and determine exposures in fair ranking \citep{oosterhuis2021}, and they have been estimated by
sampling rankings \citep{kool2019,buchholz2022}. \citet{knyazev2026} (\MPL{}) compute them by
numerical integration of the Poisson-binomial representation, derived from the Gumbel formulation,
for the top $K$ places, with $K \le 30$ in their experiments.
The window of their Gauss--Legendre rule is set from the range of the scores and grows with their
spread. They set the division recurrence aside because avoiding division ``leads to high numerical
stability,'' and measure accuracy against a $10^{7}$-sample Monte Carlo oracle. \Cref{tab:comparison} sets their
design and that of \LAQI{} against \DEICM{}, and \cref{sec:ablations} measures both quadrature
rules under our stack-ratio bound.

For the \ICM{} itself, exact calculators stop at a few dozen players, Monte Carlo simulation
\citep{streib2011} is the reference for validation, and commercial tools use undisclosed
approximations \citep{hrc2020}. In game-theoretic analyses of tournament play, the payoff of a hand
is the value of the resulting stacks: \citet{miltersen2007} compute it by solving the recursive
game, as the probability of winning the heads-up tournament, and \citet{ganzfried2008} start from
the \ICM{} and refine it by iteration.
The \ICM{} itself has been examined as a model by \citet{diaconis2022} and \citet{kim2025}.

\citet{wolters2026} (\LAQI{}) published a deterministic method on the same representation,
which was reached through the storage-and-search connection mentioned in the GTO Wizard post
\citep{tombos2024}. It returns equities only, and its comparison with exact values covers fields of
up to 13 players; larger fields are checked by self-convergence and against Monte Carlo. Its adjoint
pass obtains all equities from forward prefix states at $O(nm)$ per node, and on the exp-exp grid it
can replace the deconvolution of \cref{sec:allplayers} when only equities are needed.
A fixed grid also permits reuse of intermediate quantities, including adjoint states, that depend
only on stacks and prizes unchanged between instances, provided the instances use a common
normalization scale and satisfy the bounds used to choose the grid.

%% file: sections/07-discussion.tex
\section{Discussion and conclusion}
\label{sec:discussion}

\DEICM{} computes the placement marginals of any Plackett--Luce distribution given as positive
weights, within the stack-ratio bound its grid was set for. It therefore applies directly to the
propensities and exposures of \cref{sec:related-pl}, where the weights are exponentiated scores,
and to the depth distributions of records under the move-to-front rule and the hit probabilities
of items in LRU caches, which are rows or partial sums of the same matrix.
Its per-node cost is $O(n^{2})$ regardless of the number of places requested; when only a few
places or only the equities are needed, the prefix--suffix products or the adjoint pass have lower
asymptotic cost on the same grid (\cref{sec:algorithm}).

Two limitations remain. Probabilities that are themselves small are accurate in absolute but not in
relative terms; where relative accuracy is required in the tails, the prefix--suffix products can
replace the deconvolution for the places concerned, at a cost that grows quadratically in the number
of places (\cref{sec:loo}).
And the \ICM{} is a model, whose adequacy is a separate question \citep{kim2025}.

\paragraph{Reproducibility.}
\Cref{alg:deicm} with the grid of \cref{sec:errorbudget} and \cref{app:constants} specifies the
method, and \cref{app:references,app:suite} the reference families and the validation suite.
This description is intended to support independent reimplementation and comparison against the same
reference families. The implementation used in the paper is independent of the production code;
neither is released.

\paragraph{Conclusion.}
The placement probabilities of the \ICM{}, and of any Plackett--Luce distribution, can be
computed for all players and all places at once, for fields of thousands of players, to nearly
the accuracy of double-precision arithmetic. Two pieces of numerical design make this possible
on top of a classical integral representation. The first is a double-exponential grid fixed once
for the admissible inputs, whose truncation errors are bounded in closed form and whose step
follows from the field size; the second is a leave-one-out deconvolution whose direction rule
and clamp deliver every player's coefficients from one product with an absolute error of $O(n
\ueps)$. Both were
tested on \NumInstances{} exactly solvable instances with up to \MaxN{} players, stack ratios up
to $10^{9}$ and a wide range of prize structures, and the errors stayed at the rounding floor
throughout.

%% file: sections/A-appendix.tex
\section{Proofs}
\label{app:proofs}

\subsection[Proof of the truncation proposition]{Proof of \cref{prop:truncation}}
\label{app:truncation}

All coefficients satisfy $0 \le \loo{i}{r}(t) \le 1$ and $\sum_r \loo{i}{r}(t) = 1$. The mass
of the density $\nstack{i} e^{-\nstack{i} t}$ on $[0, t_{\mathrm{L}})$ is $1 -
e^{-\nstack{i} t_{\mathrm{L}}} \le \nstack{i} t_{\mathrm{L}}$, which bounds the omitted
contribution to every place and to their sum. For $r \ge 1$ the coefficient $\loo{i}{r}(t)$ is
the probability that exactly $r$ other clocks are below $t$, which is at most the probability
that at least one is, $\sum_{j \ne i} (1 - e^{-\nstack{j} t}) \le \widetilde{S} t$; integrating
$\nstack{i} \widetilde{S} t$ over $[0, t_{\mathrm{L}})$ gives the bound $\nstack{i}
\widetilde{S} t_{\mathrm{L}}^{2} / 2$ for places $r + 1 \ge 2$. Since $\Q{i}{1} = \nstack{i} /
\widetilde{S}$ exactly, the omitted mass $\nstack{i} t_{\mathrm{L}}$ is a relative error of
at most $\widetilde{S} t_{\mathrm{L}}$ on $\Q{i}{1}$.

The mass of the density on $[t_{\mathrm{R}}, \infty)$ is $e^{-\nstack{i} t_{\mathrm{R}}}$,
which bounds the omitted contribution to every place. For the equities, let $H_i(t) = \sum_r
\prize{r+1} \loo{i}{r}(t)$ be the expected prize of player $i$ conditioned on $\clock{i} = t$.
The number of players ahead of $i$ is a sum of independent indicators $\ind{\clock{j} < t}$,
each non-decreasing in $t$, so it is stochastically non-decreasing in $t$, and since the prizes
are non-increasing in the place, $H_i$ is non-increasing. Hence the omitted part
$\int_{t_{\mathrm{R}}}^{\infty} \nstack{i} e^{-\nstack{i} t} H_i(t)\, \mathrm{d}t$ is at most
$H_i(t_{\mathrm{R}})\, e^{-\nstack{i} t_{\mathrm{R}}}$, while the retained part is at least
$H_i(t_{\mathrm{R}}) (1 - e^{-\nstack{i} t_{\mathrm{R}}})$; their ratio is at most $2\,
e^{-\nstack{i} t_{\mathrm{R}}}$ when $e^{-\nstack{i} t_{\mathrm{R}}} \le 1/2$. \qed

\paragraph{Truncating the lattice sum.}
The proposition bounds the tails of the integral, whereas the quadrature omits lattice terms;
to bound those, write $D_i(x)=\nstack{i}e^{-\nstack{i}\dexp(x)}\dexp'(x)$.
For the normalization, window and steps of this paper, $D_i$ is increasing to the left of
$\xL$ and decreasing to the right of $\xR-\step$. Indeed, with $z=e^{-x}$ and
$a=\nstack{i}\dexp(x)$,
\[
  \frac{\mathrm{d}}{\mathrm{d}x}\log D_i(x)
  =1+z-\frac{z}{1+z}-a(1+z),
\]
which has the claimed signs in those tails. Since $0\le\loo{i}{r}\le1$, comparison
of each omitted term with the integral over its neighboring interval gives
\[
  \step\sum_{k<0} D_i(\xL+k\step)\loo{i}{r}(\dexp(\xL+k\step))
  \le \nstack{i}t_{\mathrm{L}},\qquad
  \step\sum_{k\ge0} D_i(\xR+k\step)\loo{i}{r}(\dexp(\xR+k\step))
  \le e^{-\nstack{i}\dexp(\xR-\step)}.
\]
For $r\ge1$, the increasing left-tail envelope
$\nstack{i}\widetilde{S}\dexp(x)\dexp'(x)$ gives the sharper bound
$\nstack{i}\widetilde{S}t_{\mathrm{L}}^2/2$ as in the proposition.
On the right, $D_i(x)H_i(\dexp(x))$ is also decreasing, so the relative equity bound
holds with $t_{\mathrm{R}}$ replaced by $\dexp(\xR-\step)$.
Even at the largest step $1/8$, the resulting absolute probability bound is below
$2\times10^{-17}$ for every admissible stack.

\paragraph{The last place.}
Almost all of the omitted mass belongs to the last place, and the \emph{relative} error of a
last-place probability can exceed the bound of the proposition by a factor of order $n$. For
$t \ge t_{\mathrm{R}}$ every other player $j$ satisfies $\behind{j}(t) \le
e^{-\nstack{\min} t_{\mathrm{R}}}$, so the probability that not all of them are ahead of $i$ is
at most $\sum_{j \ne i} \behind{j}(t) \le n\, e^{-\nstack{\min} t_{\mathrm{R}}}$; this is the
fraction of the omitted mass that can fall on places other than the last. Let $t_0 = \log(2n) /
\nstack{\min}$, so that $\sum_{j \ne i} \behind{j}(t) \le 1/2$ for $t \ge t_0$. Then
\[
  \Q{i}{n} \;\ge\; \int_{t_0}^{\infty} \nstack{i} e^{-\nstack{i} t} \prod_{j \ne i} \bigl( 1 -
  \behind{j}(t) \bigr)\, \mathrm{d}t
  \;\ge\; \tfrac12 \int_{t_0}^{\infty} \nstack{i} e^{-\nstack{i} t}\, \mathrm{d}t
  \;=\; \tfrac12\, e^{-\nstack{i} t_0},
\]
and the ratio of the omitted mass to this lower bound is at most $2\, e^{-\nstack{i}
(t_{\mathrm{R}} - t_0)} \le 2\, e^{-\nstack{\min} (t_{\mathrm{R}} - t_0)} = 4n\,
e^{-\nstack{\min} t_{\mathrm{R}}}$, where $\nstack{\min} \ge 1 / (c\rho)$ under
\cref{eq:ratio-bound,eq:normalization}. At the right end $\xR = 27.5$ of the window of
\cref{sec:errorbudget} this bound is $1.4 \times 10^{-15}$ for $n = 4{,}000$.

\subsection[Proof of the width lemma]{Proof of \cref{lem:width}}
\label{app:width}

Since $e^{a} - 1 = (1 - q) / q$, we have $2 g(a_j) = a_j^{2} \behind{j} / (1 - \behind{j})$.
By the Cauchy--Schwarz inequality,
\[
  \mu'^{2} = \Bigl( \sum_{j \ne i} a_j \sqrt{\tfrac{\behind{j}}{1 - \behind{j}}} \cdot
  \sqrt{\behind{j} (1 - \behind{j})} \Bigr)^{2}
  \;\le\; \sum_{j \ne i} \frac{a_j^{2} \behind{j}}{1 - \behind{j}} \cdot \sum_{j \ne i} \behind{j} (1 - \behind{j})
  \;=\; 2 \sum_{j \ne i} g(a_j) \cdot v ,
\]
with equality if and only if $a_j / (1 - \behind{j})$ is constant in $j$, that is, if and
only if all $a_j$ are equal, since $a / (1 - e^{-a})$ is strictly increasing. \qed

\subsection{Discretization on a logarithmic grid}
\label{app:strip}

The prediction of \cref{eq:kappa} rests on the Gaussian picture of the bumps. A rigorous bound with the same
exponential dependence follows from the strip theorem for the trapezoidal rule \citep[Theorem
5.1]{trefethen2014}: if $f$ is analytic in the strip $|\operatorname{Im} x| < d$, tends to zero
uniformly as $|\operatorname{Re} x| \to \infty$ in the strip, and $\int_{-\infty}^{\infty}
|f(\xi + iy)|\, \mathrm{d}\xi \le N$ for all $|y| < d$, then the trapezoidal sum $I_\step[f] =
\step \sum_{\ell \in \mathbb{Z}} f(\ell \step)$ satisfies $|I_\step[f] - \int f| \le 2N /
(e^{2\pi d / \step} - 1)$. We apply it to the integrand of place $r + 1$ of player $i$ on the
lattice uniform in $\log t$, $f(x) = f_{i,r}(x) = \nstack{i} e^{-\nstack{i} e^{x}}\, \loo{i}{r}
(e^{x})\, e^{x}$. The exp-exp grid approaches uniform spacing in $\log t$ on its logarithmic branch;
the remark at the end of this section concerns the other branch.

\begin{lemma}[One factor]
  \label[lemma]{lem:factor}
  For $a \ge 0$ and $|y| < \pi/2$, let $q = \exp(-a e^{iy})$. Then $|q| + |1 - q| \le
  \exp(\gmax \tan^{2} y)$, with $\gmax = \max_{b > 0} b^{2} / (2 (e^{b} - 1)) \approx 0.324$ as in
  \cref{lem:width}.
\end{lemma}

\begin{proof}
  Write $\varrho = |q| = e^{-a \cos y} \in (0, 1]$ and $\theta = a \sin y$, so that $q =
  \varrho e^{-i\theta}$ and $|1 - q|^{2} = (1 - \varrho)^{2} + 4\varrho \sin^{2}(\theta / 2)$.
  If $\varrho = 1$ the claim is trivial. Otherwise $\sqrt{A^{2} + B} \le A + B / (2A)$ for $A >
  0$, $B \ge 0$, and $\sin^{2}(\theta / 2) \le \theta^{2} / 4$ give
  \[
    |1 - q| \;\le\; (1 - \varrho) + \frac{\varrho\, a^{2} \sin^{2} y}{2 (1 - \varrho)} ,
  \]
  so that, with $b = a \cos y > 0$ and $\varrho = e^{-b}$,
  \[
    |q| + |1 - q| \;\le\; 1 + \tan^{2} y \cdot \frac{b^{2}}{2 (e^{b} - 1)}
    \;\le\; 1 + \gmax \tan^{2} y \;\le\; \exp(\gmax \tan^{2} y) . \qedhere
  \]
\end{proof}

The bound is sharp as $y \to 0$, since $g(a) y^{2}$ is the exact second-order term of $\log(|q|
+ |1 - q|)$. The $\ell^{1}$ norm of the coefficients of a product of polynomials is at most the
product of the $\ell^{1}$ norms of the factors, and a factor $q + (1 - q) z$ has norm $|q| + |1
- q|$, so the lemma gives $\sum_r |\loo{i}{r}(t)| \le \exp((n - 1) \gmax \tan^{2} y)$ for $t =
\tau e^{iy}$ with $\tau > 0$ and $|y| < \pi/2$.

\begin{theorem}[Discretization]
  \label{thm:strip}
  For every player $i$, every place $r + 1$, every step $\step > 0$ and every $d \in (0,
  \pi/2)$,
  \begin{equation}
    \Bigl| I_\step[f] - \int_{-\infty}^{\infty} f(x)\, \mathrm{d}x \Bigr|
    \;\le\; \frac{2 \exp\bigl( (n - 1)\, \gmax \tan^{2} d \bigr)}{\cos d \;\bigl( e^{2\pi d / \step} - 1 \bigr)} .
    \label{eq:strip}
  \end{equation}
  With $d = \pi / ((n - 1) \gmax \step) \le 1/2$, that is, for $(n - 1) \step \ge 2\pi / \gmax
  \approx 19.4$, the right-hand side is at most
  \begin{equation}
    2.3\, \exp\Bigl( -\bigl( 1 - \eta(d) \bigr) \frac{\pi^{2}}{(n - 1)\, \gmax\, \step^{2}} \Bigr)
    \Bigl( 1 - e^{-2\pi d / \step} \Bigr)^{-1},
    \qquad \eta(d) = \frac{\tan^{2} d}{d^{2}} - 1 \le 0.194 .
    \label{eq:rate}
  \end{equation}
\end{theorem}

\begin{proof}
  $f$ is entire. On the line $x = \xi + iy$ we have $e^{x} = e^{\xi} e^{iy}$, so
  $|e^{-\nstack{i} e^{x}}| = e^{-\nstack{i} e^{\xi} \cos y}$, $|e^{x}| = e^{\xi}$, and
  $|\loo{i}{r}(e^{x})| \le \exp((n - 1) \gmax \tan^{2} y)$ by the remark above. Substituting
  $\tau = e^{\xi}$,
  \[
    \int_{-\infty}^{\infty} |f(\xi + iy)|\, \mathrm{d}\xi
    \;\le\; \exp\bigl( (n - 1) \gmax \tan^{2} y \bigr) \int_{0}^{\infty} \nstack{i} e^{-\nstack{i} \tau \cos y}\, \mathrm{d}\tau
    \;=\; \frac{\exp\bigl( (n - 1) \gmax \tan^{2} y \bigr)}{\cos y},
  \]
  which is increasing in $|y|$, so $N = \exp((n - 1) \gmax \tan^{2} d) / \cos d$ serves for all
  $|y| < d$; the same bounds show that $f \to 0$ uniformly in the strip as $|\xi| \to \infty$.
  The strip theorem gives \cref{eq:strip}. For \cref{eq:rate}, the exponent at $d = \pi / ((n -
  1) \gmax \step)$ is $(n - 1) \gmax d^{2} (1 + \eta(d)) - 2\pi d / \step = -(1 - \eta(d))\,
  \pi^{2} / ((n - 1) \gmax \step^{2})$, and $2 / \cos(1/2) \le 2.3$.
\end{proof}

The theorem gives a bound of the form $A e^{-\kappa / ((n - 1) \step^{2})}$ with
$\kappa=(1-\eta(d))\pi^2/\gmax$ and
$A=2.3/(1-e^{-2\pi d/\step})$. It supports the exponential rate in the Gaussian
prediction, but does not establish its prefactor $2/n$ or an extremal input for the
quadrature error. The smaller prefactor in \cref{eq:kappa} uses the mass $1/n$ of an
equal-stack placement probability. The same bound holds for the equity
integrand $\sum_r \prize{r+1} f_{i,r}$ with the first-place prize as an extra factor. The theorem is
stated for the lattice uniform in $\log t$. Under $x \mapsto \dexp(x)$ the argument of $t$ on
the line $\operatorname{Im} x = y$ is $y + e^{-\operatorname{Re} x} \sin y$, which exceeds $y$
on the double-exponential branch, so \cref{lem:factor} does not apply there verbatim; the bumps
on that branch are those of the early places, which are broad (\cref{lem:width}), and the
compression $1 + e^{-x}$ of the map is included in the prediction of \cref{sec:errorbudget}.

\subsection[Rounding error of the full product]{Rounding error of the full product, \cref{eq:product-error}}
\label{app:product}

Write $x_j = \nstack{j} t$ and assume that $\exp$ and $\operatorname{expm1}$ are returned
with relative error at most $\ueps$. Including the rounding of $x_j$, the computed
$\behind{j} = e^{-x_j}$ has relative error at most $(x_j + 1)\ueps$, and
$1-\behind{j}$, evaluated as $-\operatorname{expm1}(-\operatorname{fl}(x_j))$, has relative
error at most $2\ueps$, since $x_j/(e^{x_j}-1) \le 1$.
The exact coefficients before a convolution step are non-negative and sum to one. Thus the errors
in these two factors contribute at most
$[(x_j+1)\behind{j}+2(1-\behind{j})]\ueps$ to the $\ell^1$ error.
The two multiplications together contribute at most $\ueps$, and the addition contributes
another $\ueps$ \citep[Chapter 3]{higham2002}. To first order, the error introduced by one
convolution step is therefore at most
\[
  \bigl[(x_j+1)\behind{j}+2(1-\behind{j})+2\bigr]\ueps
  = \bigl[4+(x_j-1)e^{-x_j}\bigr]\ueps
  \le (4+e^{-2})\ueps .
\]
Convolution with the exact factor, whose $\ell^1$ norm is one, does not amplify existing
errors. Summing over the $n$ steps gives \cref{eq:product-error}. \qed

\subsection[Proof of the stability theorem]{Proof of \cref{thm:stability}}
\label{app:stability}

Write $E_r$ for $\scaledloo{i}{r}$, $\widehat{E}_r$ for its computed value, $e_r =
\widehat{E}_r - E_r$, and $d_r = |\widehat{\coef{r}} - \coef{r}|$ for the input errors, with
$\sum_r d_r \le \theta$. All quantities below are non-negative in exact arithmetic, and we use
$\mult{i} E_{r-1} \le \coef{r}$ and $E_r \le \coef{r}$ from \cref{lem:invariant} throughout.
Terms of order $\ueps^{2}$ are dropped.

\emph{Multipliers.} Write $x = \nstack{i} t$. As in \cref{sec:fullproduct}, $\behind{i}$ and $1 -
\behind{i}$ are computed as $\exp(-\operatorname{fl}(x))$ and
$-\operatorname{expm1}(-\operatorname{fl}(x))$; the rounding of the argument and of the function
give $\widehat{\behind{i}}$ a relative error of at most $(x + 1)\ueps$ and the computed value of $1
- \behind{i}$ one of at most $2\ueps$, and a division adds $\ueps$. Hence $\widehat{\mult{i}} =
\mult{i}(1 + \theta_\alpha)$ and $\widehat{\beta_i} = \beta_i (1 + \theta_\beta)$ with
$|\theta_\alpha|, |\theta_\beta| \le (x + 4)\ueps$, where $\beta_i = 1/\mult{i} = \behind{i} / (1 -
\behind{i})$ is the multiplier of the backward recurrence, computed as a quotient in the same way
rather than by inverting $\widehat{\mult{i}}$.

\emph{Forward recurrence ($\behind{i} \ge 1/2$, so $x \le \log 2$).} One step computes
$\widehat{E}_r = \bigl( \widehat{\coef{r}} - \widehat{\mult{i}} \widehat{E}_{r-1} (1 +
\delta_1) \bigr) (1 + \delta_2)$ with $|\delta_1|, |\delta_2| \le \ueps$, so $e_r = -\mult{i}
e_{r-1} + \eta_r$ with
\[
  |\eta_r| \;\le\; d_r
  + |\theta_\alpha|\, \mult{i} E_{r-1}
  + \ueps\, \mult{i} E_{r-1}
  + \ueps\, E_r
  \;\le\; d_r + (\log 2 + 6)\ueps\, \coef{r}
  \;\le\; d_r + 7\ueps\, \coef{r} ,
\]
using $\mult{i} E_{r-1} \le \coef{r}$ and $E_r \le \coef{r}$. Unrolling, $e_r = \sum_{s \le r}
(-\mult{i})^{r-s} \eta_s$, and since $\mult{i} \le 1$ under the rule,
\[
  |e_r| \;\le\; \sum_{s \le r} |\eta_s| \;\le\; \theta + 7\ueps \sum_{s} \coef{s} \;=\; \theta +
  7\ueps ,
\]
which is \cref{eq:stability-bound}.

\emph{Backward recurrence ($\behind{i} < 1/2$, so $x > \log 2$).} Here $\beta_i < 1$. One
step computes $\widehat{E}_{r-1} = \widehat{\beta_i} (\widehat{\coef{r}} - \widehat{E}_r)(1 +
\delta_1)(1 + \delta_2)$, so $e_{r-1} = -\beta_i e_r + \eta_{r-1}$ with
\[
  |\eta_{r-1}| \;\le\; \beta_i d_r + \beta_i (\coef{r} - E_r)
  \bigl( |\theta_\beta| + |\delta_1| + |\delta_2| \bigr)
  \;\le\; d_r + (x + 6)\ueps\, E_{r-1} ,
\]
using $\beta_i < 1$ and $\beta_i (\coef{r} - E_r) = E_{r-1}$. Now $E_{r-1} \le \beta_i
\coef{r} \le 2 e^{-x} \coef{r}$, and $2 e^{-x} (x + 6) \le 7$ for $x \ge \log 2$, so
$|\eta_{r-1}| \le d_r + 7\ueps\, \coef{r}$; unrolling with $\beta_i < 1$ gives again
$|e_{r-1}| \le \theta + 7\ueps$.

\emph{Clamp and wrong direction.} The clamp of \cref{sec:clamp} replaces $\widehat{E}_r$ by its
projection onto $[0, \widehat{\coef{r}}]$; since $E_r \in [0, \coef{r}]$, the projection can
increase the error by at most $d_r$ at that step. More precisely, for any tentative value $y$,
\[
 |\operatorname{proj}_{[0,\widehat{\coef{r}}]}(y)-E_r|
 \le |y-E_r|+d_r,
\]
by non-expansiveness of projection and comparison with $[0,\coef{r}]$.
Applied at every step, this adds at most $\sum_r d_r\le\theta$ to the preceding
unrolled bounds. Thus the unclamped recurrence satisfies \cref{eq:stability-bound},
and the clamped recurrence satisfies the conservative bound
$2\theta+7\ueps+O(\ueps^2)$. If the
direction is chosen against the rule, the propagation factor has modulus $\max(\mult{i},
1/\mult{i}) > 1$ at every step, and the unrolled sum amplifies $\eta_s$ by that factor to the
power of the number of steps between $s$ and $r$.
\qed

\subsection{Rounding errors in the far tail}
\label{app:fartail}

\begin{figure}[b]
  \centering
  \IfFileExists{figures/atlas.pdf}{\includegraphics[width=0.92\linewidth]{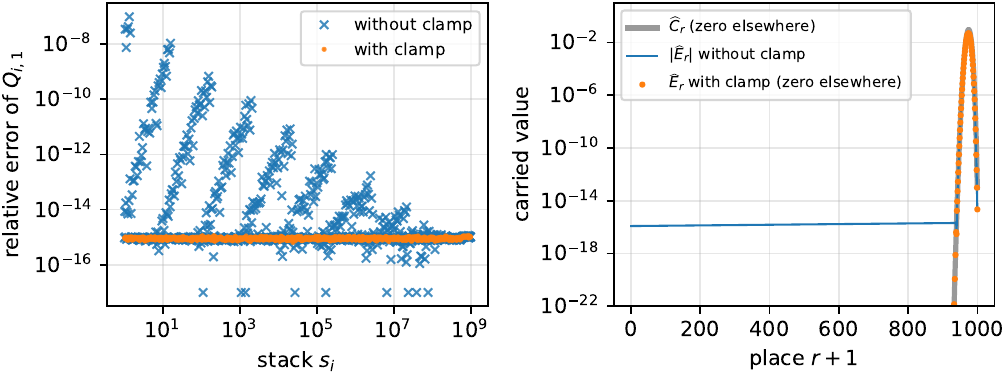}}{\todo{figure pending}}
  \caption{Left: relative errors of the first-place probabilities $\Q{i}{1} =
  \stack{i}/\stacksum$ of a $1{,}000$-player field with a geometric stack progression of ratio
  $10^{9}$, against the stack, without and with the clamp, on the grid of \cref{sec:errorbudget}.
  Right: the carried values $E_r=\scaledloo{i}{r}$ of the backward recurrence at the player--node
  pair whose residue contributes most, relative to $\Q{i}{1}$, to the errors on the left.
  Zeros are not drawn on the logarithmic scale.}
  \label{fig:atlas}
\end{figure}

To isolate the mechanism of \cref{sec:clamp}, first consider a field with one stack far
below all the others. Let $i$ be the player with the
smallest normalized stack $\nstack{\min}$ and consider a node with $t$ slightly above $\log 2
/ \nstack{\min}$, so that $\behind{i} = 1/2 - \epsilon$ and $i$ is the first player of the
backward group. Suppose the other stacks are large enough that every other
$\behind{j} = e^{-\nstack{j} t}$ has underflowed to zero. Setting those probabilities to zero, $\coef{n} = 1 - \behind{i}$, $\coef{n-1} =
\behind{i}$, and $\coef{r} = 0$ for $r < n - 1$ in exact arithmetic, and
$E_r = \scaledloo{i}{r}$ is $\behind{i}$ for $r = n-1$ and zero otherwise. The backward step
from $r = n-1$ to $r = n-2$ computes $(\widehat{\coef{n-1}} - \widehat{E}_{n-1}) / \mult{i}$,
where both terms equal $\behind{i}$ up to rounding, so the result is a residue $\eta$ with
$|\eta| \lesssim \ueps$ instead of zero; each further step multiplies the carried value by
$1/\mult{i} = \behind{i} / (1 - \behind{i}) = 1 - O(\epsilon)$, so $\widehat{E}_0 \approx \pm
\eta\, (1 - O(\epsilon))^{n}$ is still of order $\ueps$ when $n \epsilon \lesssim 1$.
The crossing of $\behind{i}$ through $1/2$ lies inside the window; whether a node falls
within $\epsilon \lesssim 1/n$ of it depends on the grid. With the clamp, the value carried
from $r=n-1$ to $r=n-2$ is projected onto $[0,\widehat{\coef{n-2}}]=\{0\}$,
and nothing propagates.

The contribution of such a node to $\widehat{Q}_{i,1}$ is $\nstack{\min} w_\ell \widehat{E}_0$
with $w_\ell = \step\, t_\ell (1 + e^{-x_\ell}) \approx \step\, t_\ell$ on the logarithmic
branch of the map and $t_\ell \approx \log 2 / \nstack{\min}$, hence of order $\step\, \ueps$
in absolute terms. The true value is $\Q{i}{1} = \nstack{\min} / \widetilde{S}$, so the
relative error is of order $\ueps \step\, \widetilde{S} / \nstack{\min}$. Here $\widetilde{S}
/ \nstack{\min} = \stacksum / s_{\min}$ lies between the stack ratio $s_{\max} / s_{\min}$ and
$n$ times it, hence at most $n\rho$ under \cref{eq:ratio-bound}. This is the law quoted in
\cref{sec:clamp}.

The geometric progression of \cref{fig:atlas} illustrates the same mechanism without
requiring all other factors to underflow: the residual propagates below the coefficient
bump, where the clamp sets it to zero. Here $\stacksum / s_{\min}
\approx 50\rho$, and with $\rho = 10^{9}$ and $\step = 1/32$ it predicts relative errors of
order $10^{-7}$, consistent with the largest errors observed without the clamp in \cref{fig:atlas}
and over the extremal set in \cref{sec:loo}.

\section{Error budget details}
\label{app:error}

\begin{table}[b]
  \centering
  \small
  \caption{Measured worst absolute error of placement probabilities over the extremal set of
  \cref{sec:accuracy}, by step and field size, on the window $[-3.75, 27.5)$; the
  probability counterpart of \cref{tab:step}. Shaded cells are at the rounding floor
  (below $5 \times 10^{-14}$). Dashes indicate untested combinations.}
  \label{tab:sweeph}
  \IfFileExists{tables/sweep_h_absQ.tex}{\input{tables/sweep_h_absQ}}{\todo{table pending}}
\end{table}

\begin{table}[b]
  \centering
  \small
  \caption{The prediction of \cref{eq:kappa} with $\kappa = 28$, for every cell of \cref{tab:sweeph}
  that is above the floor. Shaded cells are at the floor in the measurement.}
  \label{tab:sweeph-pred}
  \IfFileExists{tables/sweep_h_pred.tex}{\input{tables/sweep_h_pred}}{\todo{table pending}}
\end{table}

\Cref{tab:sweeph} gives the worst absolute error of a placement probability over the extremal
set of \cref{sec:accuracy} as a function of the step and of the field size, and
\cref{tab:sweeph-pred} the prediction of \cref{eq:kappa} with $\kappa = 28$ for every cell above
the rounding floor. The measurements are within a factor \PredRatioMin{} to \PredRatioMax{} of
the prediction wherever $(n - 1) \step \ge 20$, the regime in which the Gaussian bumps
dominate. At coarser steps in small fields, another error scale matters. On a uniform $\log t$
grid, a lone exponential density $\nstack{i} e^{-\nstack{i} t}$ has leading aliasing amplitude
$2\,|\Gamma(1+2\pi i/\step)| \sim (4\pi/\sqrt{\step}) e^{-\pi^2/\step}$ as $\step\to0$,
independent of $n$: about $5\times10^{-8}$ at $\step=1/2$ and $2\times10^{-16}$ at $1/4$.
In a field the density is also multiplied by coefficients whose factors $1 - e^{-\nstack{j} t}$ have
tails of the same exponential kind, so this figure understates the error.
The measured errors of small fields set the cap at $1/8$.

\begin{figure}[t]
  \centering
  \IfFileExists{figures/sweep_window.pdf}{\includegraphics[width=0.8\linewidth]{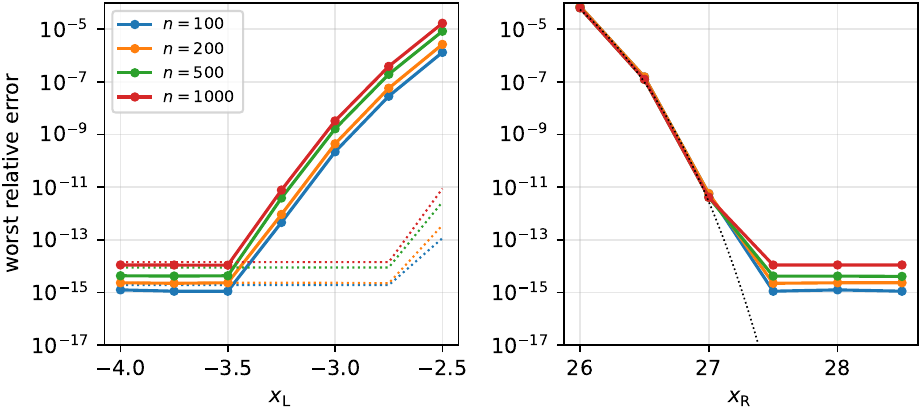}}{\todo{figure pending}}
  \caption{Worst relative error of equities over the extremal set at each size $n \le 1{,}000$ (solid)
  as each window endpoint is varied: $\xL$ in the left panel and $\xR$ in the right panel.
  The other endpoint and the step are fixed at the values of \cref{sec:errorbudget}.
  Dotted colored lines on the left: the same with the first place taken in closed form, $\Q{i}{1} =
  \stack{i} / \stacksum$. Black dotted line on the right: the omitted density mass of
  \cref{prop:truncation}, $e^{-\nstack{\min} t_{\mathrm{R}}}$, for the continuous integral.}
  \label{fig:window}
\end{figure}

\paragraph{Window.}
\Cref{fig:window} varies the window at the step of \cref{sec:errorbudget}. The error stays at
the floor while the endpoints satisfy \cref{prop:truncation} and rises as soon as they do not:
on the left by the amount $\widetilde{S} t_{\mathrm{L}}$ the proposition predicts, on the right
abruptly, as the window cuts into the density of the smallest stack. The left condition is the
one of the first place; with the first place taken in closed form the left end can move from
$-3.75$ to $-2.75$, about $3\%$ of the nodes, and the window is then set by the second place.
On the right the error approximately follows the mass $e^{-\nstack{\min} \dexp(x_{\mathrm{R}} -
\step/2)}$ omitted beyond $x_{\mathrm{R}} - \step/2$, corresponding to an effective end of the
trapezoidal sum: $4 \times 10^{-12}$ at $\xR = 27.0$ for $\step = 1/32$ and $6 \times 10^{-12}$ for
$\step = 1/16$, within the lattice bound of \cref{app:truncation}; the floor is reached at $27.5$.
The relative error of the last-place probability of a minimal stack that shares the last place with
$n - 1$ others, which can be $n$ times larger (\cref{app:truncation}), is $4 \times 10^{-9}$ at
$\xR = 27.0$ for $n = 1{,}000$ and reaches its floor at $\xR = 27.5$ as well (not shown).

\section{Numerical settings}
\label{app:constants}

The settings cover $2 \le n \le 4{,}000$, the prizes of \cref{sec:icm}, and the
stack-ratio bound of \cref{eq:ratio-bound} with $\rho = 10^{9}$. Inputs are represented in
double precision, with the prizes scaled so that the equities lie in the normal range of the format. The stacks are sorted and normalized by $20\, s_{\max}$, so that
$\nstack{i} \in [1/(20\rho), 1/20]$.

The grid is the exp-exp lattice $t = \exp(x - e^{-x})$
with the trapezoidal weights $\step\, t\, (1 + e^{-x})$ on $[\xL, \xR) = [-3.75, 27.5)$, with
$\step = 1/8$ for $2\le n \le 50$, $1/16$ for $50<n \le 200$, $1/32$ for
$200<n \le 1{,}000$ and $1/64$ for $1{,}000<n \le 4{,}000$
($\nodes = 250$, $500$, $1{,}000$, $2{,}000$).
These fixed ranges are guided by the error prediction and numerical validation.

The first place is integrated like the others. The direction rule
switches at $\behind{i} = 1/2$, the clamp interval is $[0, \widehat{\coef{r}}]$ (index shifted by
one in the backward direction), and the ahead probabilities are evaluated as
$-\operatorname{expm1}(-\nstack{j} t)$.

The reference implementation used for \cref{sec:results}
processes the nodes in blocks of $32$ (\cref{sec:implementation}) and evaluates the blocks in
parallel with a fixed assignment of nodes to threads and a fixed order of reduction.

\section{Exact reference families}
\label{app:references}

\Cref{tab:references} lists the families of \cref{sec:references}. Below, $P_{r}(X)$ denotes
the probability that after $r$ placements the set of placed players is $X$, and the
\emph{next-finisher rule} is \cref{eq:icm}: the next place goes to a remaining player $j$ with
probability $\stack{j} / \sum_{j' \notin X} \stack{j'}$. All dynamic programs start with unit probability on the empty state, iterate this
rule over $r = 0, \dots, m-1$ on a compressed state, and accumulate $\prize{r+1}$ times the
probability that the queried player is placed at step $r+1$.

\begin{table}[b]
  \centering
  \small
  \caption{Families of exactly solvable instances used as references in \cref{sec:results}.
  $g$ is the number of distinct stack values and $e$ the number of exceptional stacks.
  DP denotes dynamic programming; the references are
  evaluated in double-double arithmetic.}
  \label{tab:references}
  \begin{tabular}{@{}lll>{\raggedright\arraybackslash}p{5.6cm}@{}}
    \toprule
    Family & Stacks & Prizes & Reference \\
    \midrule
    winner-take-all & arbitrary & $m = 1$ & $\Q{i}{1} = \stack{i}/\stacksum$ \\
    two prizes & arbitrary & $m = 2$ & closed form, $O(n^{2})$ \\
    flat (satellite) & arbitrary & $\prize{r} \equiv \pi$, $m = n$ & $\V{i} = \pi$ \\
    linear & arbitrary & $\prize{r} = n - r + 1$ & $\V{i} = n + \tfrac12 - \sum_{j} \stack{j}/(\stack{i} + \stack{j})$ \citep{mccabe1965,burville1973} \\
    equal stacks & $\stack{i} \equiv s$ & arbitrary & $\V{i} = \tfrac{1}{n}\sum_r \prize{r}$ \\
    two stack sizes & $g = 2$ & arbitrary & DP on the count of placed type-1 players, $O(m \min(n_1, n_2))$ \\
    equal with exceptions & $e \le 10$ exceptional & arbitrary & DP on the set of placed exceptional players, $O(m\, e\, 2^{e})$ \\
    small fields & $n \le 20$ & arbitrary & DP on the set of placed players, $O(n\, 2^{n})$ \\
    \bottomrule
  \end{tabular}
\end{table}

\paragraph{Closed forms \citep{burville1973}.}
The first two places follow from the next-finisher rule, $\Q{i}{1} = \stack{i} / \stacksum$
and $\Q{i}{2} = \sum_{j \ne i} (\stack{j} / \stacksum) \cdot \stack{i} / (\stacksum -
\stack{j})$, in $O(n)$ per player. With flat prizes every player receives $\pi$; with equal
stacks every equity is the mean prize. For linear prizes $\V{i} = n + 1 - \Ex[\text{place of
} i]$, and under the exponential race $\Prob(\clock{j} < \clock{i}) = \stack{j} / (\stack{i}
+ \stack{j})$, so $\Ex[\text{place of } i] = 1 + \sum_{j \ne i} \stack{j} / (\stack{i} +
\stack{j})$, which yields the formula of \cref{tab:references} (the term $j = i$ contributes
$\tfrac12$); mixed over the requested record this is McCabe's expected search cost
\citep{mccabe1965,fill1996holst}.

\paragraph{Dynamic programs.}
When the stacks take $g$ distinct values with multiplicities $n_1, \dots, n_g$, the
next-finisher rule depends only on how many players of each value have been placed, so the
state after $r$ placements is a vector $(x_1, \dots, x_g)$ with $\sum_h x_h = r$.
The next placement increments $x_h$ with probability
$\stack{(h)}(n_h-x_h)/\sum_{h'}\stack{(h')}(n_{h'}-x_{h'})$. A given
player of value $h$ is still unplaced in state $x$ with probability $(n_h - x_h) / n_h$ by
exchangeability within the group, and is then placed next with probability $\stack{(h)} /
\sum_{h'} \stack{(h')} (n_{h'} - x_{h'})$, so $\Q{i}{r+1}$ is the sum of $P_r(x)$ times these
two factors over the states of layer $r$. Iterating over $r < m$ costs $O(m \min(n_1, n_2))$
for two values, the case used in the suite. When all
but $e$ players have the same stack, the state is the set $X$ of placed exceptional players
together with the count $r - |X|$ of placed regular players, and the iteration costs
$O(m\, e\, 2^{e})$. For arbitrary small fields the state is the set of placed players; the
layers $r = 0, \dots, m-1$ contain at most $2^{n}$ subsets in total, so the iteration costs $O(n\, 2^{n})$
in total. The references are evaluated in double-double arithmetic (\cref{sec:references});
in plain double precision their own rounding error, of order $10^{-12}$ for the largest
instances, would exceed the errors being measured.

\section{Validation suite}
\label{app:suite}

The suite of \cref{sec:references} uses the field sizes $n \in \{3, 10, 30, 100, 200, 500, 1{,}000,
2{,}000, 4{,}000\}$.

\smallskip
\noindent\textbf{Stacks.}
We use equal stacks; linear stacks, $\stack{i} = i$; geometric
progressions from $1$ to $10^{2}$, $10^{4}$, $10^{6}$ and $10^{9}$; Zipf, $\stack{i} \propto
i^{-1.5}$, and log-normal draws with $\sigma = 2$, both rescaled in the logarithm to ratio
$10^{9}$ when they exceed it; one stack $10^{9}$ times smaller, or larger, than the $n - 1$
others; two stack sizes, with one player, half the field, or all but one player at stack $1$
and the rest at $3$, $10^{3}$ or $10^{9}$; and ten exceptional stacks among regular stacks of
$1$: spread over the ratio $\rho$, $\rho^{-1/2 + k/9}$ for $k = 0, \dots, 9$; all below the
regular stack, $\rho^{-1/2 + k/20}$; and all above, $\rho^{1/20 + k/20}$.

\smallskip
\noindent\textbf{Prizes.}
The structures are winner-take-all; two prizes with $\prize{2} / \prize{1} = 0.5$ and $10^{-6}$; flat over all $n$
places; linear, $\prize{r} = n - r + 1$; a realistic payout $\prize{r} \propto r^{-0.8}$ and a
top-heavy payout $\prize{r} = 0.7^{\,r-1}$, both for the top $\lceil 0.15n \rceil$ places; and
satellites paying the top half of the field, the top $80\%$, or all but the last player,
equally. Each stack generator is paired with the prize structures for which
\cref{tab:references} supplies an exact reference, and the small-field family ($n \le 20$)
with all of them.

\smallskip
\noindent\textbf{Runtime measurements.}
These use log-normal stacks with $\sigma = 1$, ratio capped at $10^{6}$, the full prize vector or
the top $0.15n$ places paid, at $n = 100$, $1{,}000$, $2{,}000$ and $4{,}000$, and report the
median of three runs.

\smallskip
\noindent\textbf{Random draws.}
The log-normal draws come from a 64-bit
Mersenne twister seeded with $12345 + n$ in the suite and $99 + n$ in the runtime measurements,
passed through the normal distribution of the C++ standard library; that transform is not fixed
by the standard, so the draws are reproduced exactly only with the GNU library used here, and in
distribution with any other.

\smallskip
\noindent\textbf{Extremal set.}
The sweeps of \cref{sec:results} run on $25$ instances per field size, $n = 100$ to $4{,}000$,
chosen from the error budget. For the discretization error: equal stacks with the full matrix
and with prize steps at $80\%$ and $50\%$ of the field, at a realistic payout and at the
bubble; two-stack-size fields whose largest equal group has $n - 1$ or $n / 2$ members at ratio
$10^{3}$; and, because the relative error of an equity can be large for a player whose equity is
a single narrow transition of the group, a tiny stack, or the ten exceptions below the group,
among equal stacks under the $80\%$ satellite. For the truncation error: the geometric
progression of ratio $10^{9}$, a single outlier of ratio $10^{9}$ in either direction,
two-stack-size fields at ratio $10^{9}$ and ten exceptions spread over the ratio, with the
prizes that expose the first, the second and the last place.

At twice the step of \cref{sec:errorbudget}, the largest absolute probability errors over
the full suite and the extremal set agree exactly for all six tested field sizes. The
largest relative equity errors agree to within $0.2\%$. In the full suite, the largest
absolute probability error is attained by equal stacks with flat prizes. The largest
relative equity error is attained under the $80\%$ satellite by equal stacks at $n=100$,
and by equal stacks with a tiny stack or ten smaller exceptions at larger $n$.

%% file: tables/sweep_h_absQ.tex
\begin{tabular}{@{}rrcccccc@{}}
\toprule
$1/\step$ & $\nodes$ & $n=100$ & $n=200$ & $n=500$ & $n=1{,}000$ & $n=2{,}000$ & $n=4{,}000$ \\
\midrule
8 & 250 & $4.4 \times 10^{-10}$ & $1.1 \times 10^{-6}$ & $1.0 \times 10^{-4}$ & $3.6 \times 10^{-4}$ & -- & -- \\
12 & 375 & \cellcolor{floorshade}$6.7 \times 10^{-16}$ & $1.7 \times 10^{-11}$ & $1.1 \times 10^{-6}$ & $3.2 \times 10^{-5}$ & $1.4 \times 10^{-4}$ & -- \\
16 & 500 & \cellcolor{floorshade}$4.4 \times 10^{-16}$ & \cellcolor{floorshade}$7.8 \times 10^{-16}$ & $1.9 \times 10^{-9}$ & $1.3 \times 10^{-6}$ & $2.6 \times 10^{-5}$ & $9.1 \times 10^{-5}$ \\
24 & 750 & \cellcolor{floorshade}$6.7 \times 10^{-16}$ & \cellcolor{floorshade}$1.0 \times 10^{-15}$ & \cellcolor{floorshade}$2.3 \times 10^{-15}$ & $1.4 \times 10^{-10}$ & $2.6 \times 10^{-7}$ & $8.1 \times 10^{-6}$ \\
32 & 1{,}000 & \cellcolor{floorshade}$4.4 \times 10^{-16}$ & \cellcolor{floorshade}$7.8 \times 10^{-16}$ & \cellcolor{floorshade}$1.9 \times 10^{-15}$ & \cellcolor{floorshade}$3.2 \times 10^{-15}$ & $4.3 \times 10^{-10}$ & $3.3 \times 10^{-7}$ \\
48 & 1{,}500 & \cellcolor{floorshade}$3.3 \times 10^{-16}$ & \cellcolor{floorshade}$6.3 \times 10^{-16}$ & \cellcolor{floorshade}$1.5 \times 10^{-15}$ & \cellcolor{floorshade}$3.0 \times 10^{-15}$ & \cellcolor{floorshade}$6.4 \times 10^{-15}$ & $3.4 \times 10^{-11}$ \\
64 & 2{,}000 & \cellcolor{floorshade}$2.2 \times 10^{-16}$ & \cellcolor{floorshade}$4.5 \times 10^{-16}$ & \cellcolor{floorshade}$1.2 \times 10^{-15}$ & \cellcolor{floorshade}$2.6 \times 10^{-15}$ & \cellcolor{floorshade}$5.9 \times 10^{-15}$ & \cellcolor{floorshade}$1.1 \times 10^{-14}$ \\
\bottomrule
\end{tabular}

%% file: tables/sweep_h_pred.tex
\begin{tabular}{@{}rrcccccc@{}}
\toprule
$1/\step$ & $\nodes$ & $n=100$ & $n=200$ & $n=500$ & $n=1{,}000$ & $n=2{,}000$ & $n=4{,}000$ \\
\midrule
8 & 250 & $2.8 \times 10^{-10}$ & $1.2 \times 10^{-6}$ & $1.1 \times 10^{-4}$ & $3.3 \times 10^{-4}$ & -- & -- \\
12 & 375 & \cellcolor{floorshade}-- & $1.6 \times 10^{-11}$ & $1.2 \times 10^{-6}$ & $3.5 \times 10^{-5}$ & $1.3 \times 10^{-4}$ & -- \\
16 & 500 & \cellcolor{floorshade}-- & \cellcolor{floorshade}-- & $2.3 \times 10^{-9}$ & $1.5 \times 10^{-6}$ & $2.8 \times 10^{-5}$ & $8.3 \times 10^{-5}$ \\
24 & 750 & \cellcolor{floorshade}-- & \cellcolor{floorshade}-- & \cellcolor{floorshade}-- & $1.9 \times 10^{-10}$ & $3.1 \times 10^{-7}$ & $8.9 \times 10^{-6}$ \\
32 & 1{,}000 & \cellcolor{floorshade}-- & \cellcolor{floorshade}-- & \cellcolor{floorshade}-- & \cellcolor{floorshade}-- & $5.9 \times 10^{-10}$ & $3.8 \times 10^{-7}$ \\
48 & 1{,}500 & \cellcolor{floorshade}-- & \cellcolor{floorshade}-- & \cellcolor{floorshade}-- & \cellcolor{floorshade}-- & \cellcolor{floorshade}-- & $4.9 \times 10^{-11}$ \\
\bottomrule
\end{tabular}

%% file: refs.bib
@article{harville1973,
  author  = {Harville, David A.},
  title   = {Assigning Probabilities to the Outcomes of Multi-Entry Competitions},
  journal = {Journal of the American Statistical Association},
  volume  = {68},
  number  = {342},
  pages   = {312--316},
  year    = {1973},
  optdoi    = {10.1080/01621459.1973.10482425}
}

@book{malmuth1987,
  author    = {Malmuth, Mason},
  title     = {Gambling Theory and Other Topics},
  publisher = {Two Plus Two Publishing},
  address   = {Henderson, NV},
  year      = {1987}
}

@article{diaconis2022,
  author  = {Diaconis, Persi and Ethier, Stewart N.},
  title   = {Gambler's Ruin and the {ICM}},
  journal = {Statistical Science},
  volume  = {37},
  number  = {3},
  pages   = {289--305},
  year    = {2022},
  optdoi    = {10.1214/21-STS826}
}

@inproceedings{kim2025,
  author    = {Kim, Juho},
  title     = {Empirical Validation of the Independent Chip Model},
  booktitle = {Proceedings of the {IEEE} Conference on Games ({CoG})},
  address   = {Lisbon, Portugal},
  year      = {2025},
  note      = {arXiv:2506.00180}
}

@misc{tombos2024,
  author       = {Boshoff, Tom and Beardsell, Philippe and Inariba, Wataru},
  title        = {Theoretical Breakthroughs in {ICM}},
  year         = {2024},
  howpublished = {GTO Wizard blog},
  url          = {https://blog.gtowizard.com/theoretical-breakthroughs-in-icm/},
  note         = {Published under the byline Tombos21. Accessed September 2026}
}

@misc{hrc2020,
  author       = {{HoldemResources}},
  title        = {High Accuracy {ICM} Calculations for Large Fields},
  year         = {2020},
  howpublished = {HoldemResources.net blog},
  url          = {https://www.holdemresources.net/blog/high-accuracy-mtt-icm/},
  note         = {Accessed September 2026}
}

@misc{streib2011,
  author       = {Streib, Tysen},
  title        = {New Algorithm to Calculate {ICM} for Large Tournaments},
  year         = {2011},
  howpublished = {Two Plus Two Poker Forums, Poker Theory},
  url          = {https://forumserver.twoplustwo.com/15/poker-theory-amp-gto/new-algorithm-calculate-icm-large-tournaments-1098489/},
  note         = {Accessed September 2026}
}

@misc{wolters2026,
  author       = {Wolters, Derek},
  title        = {Log-Age Quadrature {ICM}: Fast Deterministic Dollar Equity Estimates for Poker Tournaments},
  year         = {2026},
  month        = jul,
  howpublished = {Preprint, version 1.0.0, Zenodo},
  doi          = {10.5281/zenodo.21610089},
  url          = {https://zenodo.org/records/21610089},
  note         = {Software: \url{https://github.com/curesive/log-age-quadrature-icm}}
}

@book{luce1959,
  author    = {Luce, R. Duncan},
  title     = {Individual Choice Behavior: A Theoretical Analysis},
  publisher = {Wiley},
  address   = {New York},
  year      = {1959}
}

@article{plackett1975,
  author  = {Plackett, Robin L.},
  title   = {The Analysis of Permutations},
  journal = {Journal of the Royal Statistical Society: Series C (Applied Statistics)},
  volume  = {24},
  number  = {2},
  pages   = {193--202},
  year    = {1975},
  optdoi  = {10.2307/2346567}
}

@article{yellott1977,
  author  = {Yellott, Jr., John I.},
  title   = {The Relationship between {L}uce's Choice Axiom, {T}hurstone's Theory of Comparative Judgment, and the Double Exponential Distribution},
  journal = {Journal of Mathematical Psychology},
  volume  = {15},
  number  = {2},
  pages   = {109--144},
  year    = {1977}
}

@inproceedings{kool2019,
  author    = {Kool, Wouter and van Hoof, Herke and Welling, Max},
  title     = {Stochastic Beams and Where to Find Them: The {G}umbel-Top-$k$ Trick for Sampling Sequences Without Replacement},
  booktitle = {Proceedings of the 36th International Conference on Machine Learning ({ICML})},
  series    = {Proceedings of Machine Learning Research},
  volume    = {97},
  pages     = {3499--3508},
  publisher = {PMLR},
  year      = {2019}
}

@inproceedings{oosterhuis2021,
  author    = {Oosterhuis, Harrie},
  title     = {Computationally Efficient Optimization of {P}lackett-{L}uce Ranking Models for Relevance and Fairness},
  booktitle = {Proceedings of the 44th International {ACM} {SIGIR} Conference on Research and Development in Information Retrieval},
  pages     = {1023--1032},
  publisher = {ACM},
  year      = {2021},
  optdoi    = {10.1145/3404835.3462830}
}

@inproceedings{buchholz2022,
  author    = {Buchholz, Alexander and Lichtenberg, Jan Malte and Di Benedetto, Giuseppe and Stein, Yannik and Bellini, Vito and Ruffini, Matteo},
  title     = {Low-Variance Estimation in the {P}lackett-{L}uce Model via Quasi-{M}onte {C}arlo Sampling},
  booktitle = {{SIGIR} 2022 Workshop on Reaching Efficiency in Neural Information Retrieval ({ReNeuIR})},
  year      = {2022},
  note      = {arXiv:2205.06024}
}

@inproceedings{knyazev2026,
  author    = {Knyazev, Norman and Oosterhuis, Harrie},
  title     = {Sample-Free Almost-Exact Estimation of {P}lackett-{L}uce Propensities for Off-Policy Ranking},
  booktitle = {Advances in Information Retrieval: 48th European Conference on Information Retrieval ({ECIR} 2026), Proceedings, Part I},
  series    = {Lecture Notes in Computer Science},
  volume    = {16483},
  pages     = {163--177},
  publisher = {Springer},
  year      = {2026},
  optdoi    = {10.1007/978-3-032-21289-4_11}
}

@article{mccabe1965,
  author  = {McCabe, John},
  title   = {On Serial Files with Relocatable Records},
  journal = {Operations Research},
  volume  = {13},
  number  = {4},
  pages   = {609--618},
  year    = {1965},
  optdoi  = {10.1287/opre.13.4.609}
}

@article{hendricks1972,
  author  = {Hendricks, W. J.},
  title   = {The Stationary Distribution of an Interesting {M}arkov Chain},
  journal = {Journal of Applied Probability},
  volume  = {9},
  number  = {1},
  pages   = {231--233},
  year    = {1972}
}

@article{burville1973,
  author  = {Burville, P. J. and Kingman, J. F. C.},
  title   = {On a Model for Storage and Search},
  journal = {Journal of Applied Probability},
  volume  = {10},
  number  = {3},
  pages   = {697--701},
  year    = {1973},
  optdoi  = {10.2307/3212792}
}

@article{flajolet1992,
  author  = {Flajolet, Philippe and Gardy, Dani{\`e}le and Thimonier, Lo{\"y}s},
  title   = {Birthday Paradox, Coupon Collectors, Caching Algorithms and Self-Organizing Search},
  journal = {Discrete Applied Mathematics},
  volume  = {39},
  number  = {3},
  pages   = {207--229},
  year    = {1992},
  optdoi  = {10.1016/0166-218X(92)90177-C}
}

@article{fill1996holst,
  author  = {Fill, James Allen and Holst, Lars},
  title   = {On the Distribution of Search Cost for the Move-to-Front Rule},
  journal = {Random Structures \& Algorithms},
  volume  = {8},
  number  = {3},
  pages   = {179--186},
  year    = {1996},
  optdoi  = {10.1002/(SICI)1098-2418(199605)8:3<179::AID-RSA2>3.0.CO;2-V}
}

@article{fill1996limits,
  author  = {Fill, James Allen},
  title   = {Limits and Rates of Convergence for the Distribution of Search Cost under the Move-to-Front Rule},
  journal = {Theoretical Computer Science},
  volume  = {164},
  number  = {1--2},
  pages   = {185--206},
  year    = {1996}
}

@article{jelenkovic1999,
  author  = {Jelenkovi{\'c}, Predrag R.},
  title   = {Asymptotic Approximation of the Move-to-Front Search Cost Distribution and Least-Recently Used Caching Fault Probabilities},
  journal = {The Annals of Applied Probability},
  volume  = {9},
  number  = {2},
  pages   = {430--464},
  year    = {1999},
  optdoi  = {10.1214/aoap/1029962750}
}

@misc{long2026,
  author       = {Long, Christopher D.},
  title        = {Radial Extremality for {LRU} Caching and the {F}ill--{H}olst Conjecture},
  year         = {2026},
  howpublished = {arXiv:2605.26107 [math.PR]}
}

@article{chen1997,
  author  = {Chen, Sean X. and Liu, Jun S.},
  title   = {Statistical Applications of the {P}oisson-Binomial and Conditional {B}ernoulli Distributions},
  journal = {Statistica Sinica},
  volume  = {7},
  number  = {4},
  pages   = {875--892},
  year    = {1997}
}

@article{hong2013,
  author  = {Hong, Yili},
  title   = {On Computing the Distribution Function for the {P}oisson Binomial Distribution},
  journal = {Computational Statistics \& Data Analysis},
  volume  = {59},
  pages   = {41--51},
  year    = {2013},
  optdoi  = {10.1016/j.csda.2012.10.006}
}

@article{biscarri2018,
  author  = {Biscarri, William and Zhao, Sihai Dave and Brunner, Robert J.},
  title   = {A Simple and Fast Method for Computing the {P}oisson Binomial Distribution Function},
  journal = {Computational Statistics \& Data Analysis},
  volume  = {122},
  pages   = {92--100},
  year    = {2018},
  optdoi  = {10.1016/j.csda.2018.01.007}
}

@article{peters1971,
  author  = {Peters, G. and Wilkinson, J. H.},
  title   = {Practical Problems Arising in the Solution of Polynomial Equations},
  journal = {Journal of the Institute of Mathematics and its Applications},
  volume  = {8},
  number  = {1},
  pages   = {16--35},
  year    = {1971},
  optdoi  = {10.1093/imamat/8.1.16}
}

@article{takahasi1974,
  author  = {Takahasi, Hidetosi and Mori, Masatake},
  title   = {Double Exponential Formulas for Numerical Integration},
  journal = {Publications of the Research Institute for Mathematical Sciences},
  volume  = {9},
  number  = {3},
  pages   = {721--741},
  year    = {1974},
  optdoi  = {10.2977/prims/1195192451}
}

@article{mori2001,
  author  = {Mori, Masatake and Sugihara, Masaaki},
  title   = {The Double-Exponential Transformation in Numerical Analysis},
  journal = {Journal of Computational and Applied Mathematics},
  volume  = {127},
  number  = {1--2},
  pages   = {287--296},
  year    = {2001},
  optdoi  = {10.1016/S0377-0427(00)00501-X}
}

@article{trefethen2014,
  author  = {Trefethen, Lloyd N. and Weideman, J. A. C.},
  title   = {The Exponentially Convergent Trapezoidal Rule},
  journal = {SIAM Review},
  volume  = {56},
  number  = {3},
  pages   = {385--458},
  year    = {2014},
  optdoi  = {10.1137/130932132}
}

@book{higham2002,
  author    = {Higham, Nicholas J.},
  title     = {Accuracy and Stability of Numerical Algorithms},
  edition   = {Second},
  publisher = {SIAM},
  address   = {Philadelphia},
  year      = {2002},
  optdoi    = {10.1137/1.9780898718027}
}

@inproceedings{miltersen2007,
  author    = {Miltersen, Peter Bro and S{\o}rensen, Troels Bjerre},
  title     = {A Near-Optimal Strategy for a Heads-Up No-Limit {T}exas {H}old'em Poker Tournament},
  booktitle = {Proceedings of the 6th International Joint Conference on Autonomous Agents and Multiagent Systems ({AAMAS})},
  pages     = {191:1--191:8},
  publisher = {ACM},
  year      = {2007},
  optdoi    = {10.1145/1329125.1329357}
}

@inproceedings{ganzfried2008,
  author    = {Ganzfried, Sam and Sandholm, Tuomas},
  title     = {Computing an Approximate Jam/Fold Equilibrium for 3-Player No-Limit {T}exas {H}old'em Tournaments},
  booktitle = {Proceedings of the 7th International Conference on Autonomous Agents and Multiagent Systems ({AAMAS})},
  pages     = {919--925},
  publisher = {IFAAMAS},
  year      = {2008}
}

@inproceedings{zinkevich2007,
  author    = {Zinkevich, Martin and Johanson, Michael and Bowling, Michael and Piccione, Carmelo},
  title     = {Regret Minimization in Games with Incomplete Information},
  booktitle = {Advances in Neural Information Processing Systems 20 ({NIPS} 2007)},
  pages     = {1729--1736},
  year      = {2007}
}
